\documentclass[12pt]{article}
\usepackage{amsfonts,amssymb,amsbsy,amsmath,amsthm}

\newtheorem{theorem}{Theorem}[section]
\newtheorem{lemma}{Lemma}[section]
\newtheorem{follow}{Corollary}[section]
\newtheorem{pr}{Proposition}[section]
\theoremstyle{definition}

\newcommand{\bel}{\begin{equation} \label}
\newcommand{\ee}{\end{equation}}

\newcommand{\re}{{\mathbb R}}

\def\beq{\begin{equation}}
\def\eeq{\end{equation}}

{

\begin{document}
\begin{center}
{\Large \bf Spectral Properties of a Magnetic Quantum Hamiltonian
on a Strip}\\
\end{center}

\medskip

\begin{center}
{\sc Philippe Briet, Georgi  Raikov, Eric Soccorsi}\\
\end{center}

\medskip

\begin{center}
{\em \large Dedicated to the memory of Volodya Geyler (1943 -
2007)}
\end{center}

\begin{abstract}
We consider a 2D Schr\"odinger operator $H_0$ with constant
magnetic field,  on a strip of finite width. The spectrum of $H_0$
is absolutely continuous, and contains a discrete set of
thresholds. We perturb $H_0$ by an electric potential $V$ which
decays in a suitable sense at infinity, and study the spectral
properties of the perturbed operator $H = H_0 + V$. First, we
establish a Mourre estimate, and as a corollary prove that the
singular continuous spectrum of $H$ is empty, and any compact
subset of the complement of the threshold set may contain at most
a finite set of eigenvalues of $H$, each of them having a finite
multiplicity. Next, we introduce the Krein spectral shift function
(SSF) for the operator pair $(H,H_0)$. We show that this SSF is
bounded on any compact subset of the complement of the threshold
set, and is continuous away from the threshold set and the
eigenvalues of $H$. The main results of the article concern the
asymptotic behaviour of the SSF at the thresholds, which is
described in terms of the SSF for a pair of
effective Hamiltonians. \\



\end{abstract}

\section{Introduction}
\setcounter{equation}{0}
    In the present article we consider a 2D Schr\"odinger
    operator $H_0$ with constant magnetic field $b>0$ defined on a
    strip ${\mathcal S}_L$ of width $2L$. The spectrum of $H_0$ is absolutely continuous,  equals
    the interval $[{\mathcal E}_1, \infty)$ with ${\mathcal E}_1 > 0$, and contains a countable set
    of thresholds ${\mathcal Z}$. This model is related to
    some aspects of the quantum Hall effect (see e.g. \cite{hall}, \cite{GS}). We
    perturb $H_0$ by an electric potential $V$ which decays in a
    suitable sense at infinity, and study some basic spectral
    properties of the perturbed operator $H$. First we establish
    a Mourre estimate (see \cite{Mo}) with an appropriate conjugate
    operator, and as a consequence we show that the singular
    continuous spectrum of $H$ is empty, and any compact subset of
    $\re\setminus{\mathcal Z}$ may contain at most a finite number
    of eigenvalues of $H$, each of them having a finite
    multiplicity. Similar Mourre estimates for other magnetic Hamiltonians have been obtained in
    \cite{dBP} and \cite[Chapter 3]{GL}.\\ Further, we introduce the Krein spectral shift
    function (SSF) for the operator pair $(H,H_0)$ and prove that it is
    bounded on every compact subset of $\re\setminus{\mathcal Z}$,
    and is continuous on $\re\setminus({\mathcal Z} \cup
    \sigma_p(H))$ where $\sigma_p(H)$ is the set of the eigenvalues
    of $H$. The main results of the article concern the asymptotic
    behaviour of the SSF near the thresholds of the spectrum of
    $H_0$. We show that this asymptotic behaviour is similar to
    the asymptotics near the origin of the SSF for a pair of
    effective Hamiltonians which are 1D Schr\"odinger operators.
    As a corollary we show that if the decay rate $\alpha$ of
$V$ is on the interval $(1,2)$, then the SSF has a
    singularity at each threshold, and describe explicitly the
    leading term of this singularity; if $\alpha > 2$,
    then the SSF remains bounded at the thresholds. The
    threshold
    behaviour of the SSF for a pair of 3D Schr\"odinger operators with constant
    magnetic fields has been investigated in
    \cite{FR} (see also \cite{R1}). In that case the thresholds coincide
    with the Landau levels, and the
    threshold singularities of the SSF have different nature,
    related to the spectral properties of compact Berezin-Toeplitz operators. \\
    The paper is organized as follows. In Section 2 we introduce
    some basic notations, describe the operators $H_0$ and $H$,
formulate our main results, and briefly comment on them.
    Section 3 contains the proof of our results related to the Mourre
    estimates, while the proofs of the results concerning the SSF
    can be found in Section 4.

\section{Main Results}
\setcounter{equation}{0}
    {\bf 2.1.} In this subsection we introduce some
basic notations used
 throughout the section.\\Let $X_1$, $X_2$ be two
 Hilbert spaces.\footnote{All the Hilbert spaces considered in the article are supposed to be separable.}
  We denote by ${\mathcal B}(X_1,X_2)$
 (resp., by $S_{\infty}(X_1,X_2))$ the class of bounded
 (resp., compact) operators $T : X_1 \to X_2$. Further,
 we denote by $S_p(X_1, X_2)$, $p \in$ $[1,\infty)$, the
Schatten-von Neumann class of compact operators $T : X_1 \to X_2$
for which the norm $\|T\|_p$ : = $\left({\rm
Tr}\,|T|^p\right)^{1/p}$ is finite (see e.g. \cite{S}). In this
paper we will use only the trace class $S_1$ and the
Hilbert-Schmidt class $S_2$. If $X_1 = X_2 = X$ we write
${\mathcal B}(X)$
 or $S_p(X)$ instead of ${\mathcal B}(X,X)$
 or $S_p(X,X)$, $p \in [1,\infty]$. Also, if the indication of the Hilbert
 space(s) where the corresponding operators act is irrelevant, we
 omit it in the notations of the classes ${\mathcal B}$ and $S_p$, $p
 \in [1,\infty]$. \\
  Let $T=T^*$.
We denote by ${\mathbb P}_{\mathcal O}(T)$ the spectral projection
of $T$ associated with  the Borel set ${\mathcal O} \subset
\re$.\\
Finally, if $T \in {\mathcal B}(X)$, we define the self-adjoint
operators ${\rm Re}\;T : = \frac{1}{2} (T+ T^*)$ and ${\rm Im}\;T
: = \frac{1}{2i}
(T- T^*)$.\\

{\bf 2.2.} In this subsection we introduce the operators $H_0$ and
$H$, and summarize some of their spectral properties which will
play a crucial role in the sequel.\\
 For $L>0$ put $I_L = (-L,L)$, ${\mathcal S} =
I_L \times \re$. Let
$$
H_0 : = -\frac{\partial^2}{\partial x^2} + \left(-i\frac{\partial}{\partial y}
- bx\right)^2
$$
be the 2D Schr\"odinger operator with constant scalar magnetic
field $b>0$, defined on $\{u \in {\rm H}^2({\mathcal S}_L) \; | \;
u_{|\partial {\mathcal S}_L} = 0\}$ where ${\rm H}^2({\mathcal
S}_L)$ denotes the second-order Sobolev space on ${\mathcal S}_L$.
Then we have
$$
{\mathcal F} H_0 {\mathcal F}^* = \int_{\re}^\oplus \hat{H}(k) dk,
$$
where  ${\mathcal F}$ is the partial Fourier transform with
respect to $y$, i.e.
$$
({\mathcal F}u)(x,k): = \frac{1}{\sqrt{2\pi}} \int_{\re} e^{-iyk}
u(x,y)dy, \quad (x,k) \in {\mathcal S}_L,
$$
and
$$
\hat{H}(k) : = -\frac{d^2}{dx^2} + (bx-k)^2, \quad k \in \re,
$$
is the operator defined on $D(\hat{H}): = \left\{w \in {\rm
H}^2(I_L) | w(-L) = w(L) = 0\right\}$. In what follows, we will
consider $D(\hat{H})$ as a Hilbert space equipped with the
standard scalar product of ${\rm H}^2(I_L)$. \\
 The spectrum $\sigma(\hat{H}(k))$ of the operator $\hat{H}(k)$, $k
\in \re$, is discrete and simple. Let $\left\{E_j(k
)\right\}_{j=1}^{\infty}$ be the increasing sequence of the
eigenvalues of $\hat{H}(k)$, which are even real analytic
functions of $k \in \re$ (see \cite{K}). The minimax
principle easily implies
    \bel{a2} E_j(k) = k^2 (1 +
o(1)), \quad k \to \pm \infty.
    \ee
By \cite[Theorem 2]{GS} we have
    \bel{m21}
k E_j'(k) > 0, \quad k \neq 0,
    \ee
    \bel{a5}
E_j(k) = {\mathcal E}_j + \mu_j k^2 + O(k^4), \quad k \to
    0,
    \ee
    with
    \bel{F1}
    {\mathcal E}_j : = E_j(0) > (2j-1)b, \quad \mu_j : = \frac{1}{2}
    E_j''(0) > 0.
    \ee
Thus $\sigma(H_0) = \sigma_{\rm ac}(H_0) = [{\mathcal E}_1,
\infty)$, and ${\mathcal E}_j$, $j \in {\mathbb N} : =
\{1,2,\ldots\}$, are thresholds in $\sigma(H_0)$. Set ${\mathcal
Z}: =
\bigcup_{j \in {\mathbb N}}\left\{{\mathcal E}_j\right\}$. \\
Let $V: S_L \to \re$ be an electric potential such that the
operator $|V|^{1/2} H_0^{-1/2}$ is compact.  We define the
perturbed operator $H: = H_0 + V$ as a sum in the sense of the
quadratic forms. Then we have $\sigma_{\rm ess}(H) = \sigma_{\rm
ess}(H_0) = \sigma(H_0) = [{\mathcal E}_1,\infty)$.\\

{\bf 2.3.} In this subsection we formulate our result concerning
the absence of singular continuous spectrum of the operator $H$, and some
generic properties of its eigenvalues.
\begin{theorem} \label{t1}
(i) Assume
    \bel{f4a}
VH_0^{-1} \in S_{\infty},
    \ee
    \bel{f4}
 H_0^{-1} y \frac{\partial V}{\partial
y}H_0^{-1} \in S_{\infty}.
    \ee
Then  any compact subinterval of $\re \setminus {\mathcal Z}$ may
contain at most a finite number of eigenvalues, each of them
having a finite multiplicity.\\
(ii) Suppose moreover
    \bel{f5a}
    H_0^{-1/2} y \frac{\partial V}{\partial y}H_0^{-1} \in {\mathcal B},
    \ee
    \bel{f5}
H_0^{-1} y^2 \frac{\partial^2 V}{\partial y^2}H_0^{-1} \in
{\mathcal B}.
    \ee
Then $\sigma_{\rm sc}(H) = \emptyset$.
\end{theorem}
The proof of Theorem \ref{t1} is contained in Section 3. \\
{\em Remark}: Let $U: {\mathcal S}_L \to [0,\infty)$, and let
$\Delta_D$ be the Dirichlet Laplacian on ${\mathcal S}_L$. The
Sobolev embedding theorems imply that the inclusion $U^{1/2}
(-\Delta_D)^{-1/2}$ $\in$ ${\mathcal B}$ (resp., $U^{1/2}
(-\Delta_D)^{-1/2}$ $\in$ $S_{\infty}$) is ensured by $U \in
L^q({\mathcal S}_L) + L^{\infty}({\mathcal S}_L)$ (resp., $U \in
L^q({\mathcal S}_L) + L^{\infty}_{\varepsilon}({\mathcal S}_L)$,
i.e. for each $\varepsilon > 0$ we have $U = U_1 + U_2$ with $U_1
\in L^q({\mathcal S}_L)$, $U_2 \in L^{\infty}({\mathcal S}_L)$,
$\|U_2\|_{L^{\infty}({\mathcal S}_L)} \leq \varepsilon$), $q
> 1$. Similarly,  the condition $U \Delta_D^{-1} \in {\mathcal B}$
(resp., $U \Delta_D^{-1} \in S_{\infty}$) follows from $U \in
L^2({\mathcal S}_L) + L^{\infty}({\mathcal S}_L)$ (resp., $U \in
L^2({\mathcal S}_L) + L^{\infty}_{\varepsilon}({\mathcal S}_L)$).
On the other hand, by the diamagnetic inequality  (see e.g.
\cite[Chapter 2]{S}), we have   $\|U^{\gamma} H_0^{-\gamma}\| \leq
\|U^{\gamma} (-\Delta_D)^{-\gamma}\|$, $\gamma > 0$, and,
moreover,  $U^{\gamma} (-\Delta_D)^{-\gamma} \in S_{\infty}$
entails $U^{\gamma} H_0^{-\gamma} \in S_{\infty}$. These facts
could be used in order to deduce sufficient
conditions which guarantee the validity of the hypotheses of Theorem \ref{t1}.\\

    {\bf 2.4.} This subsection contains our results on the
    threshold behaviour of the
spectral shift function for the operator pair $(H,H_0)$.  Let us
recall the abstract setting for the SSF.
   Let ${\mathcal H}_0$ and ${\mathcal H}$ be two
lower-bounded self-adjoint operators acting in the same Hilbert
space. Assume  that for some $\gamma>0$, and $E_0 <
\inf{\sigma({\cal H}_0) \cup \sigma({\cal H})}$, we have
\begin{equation} \label{10}
({\cal H} - E_0)^{-\gamma} - ({\cal H}_0 - E_0)^{-\gamma} \in S_1.
\end{equation}
Then there exists a unique $\xi(\cdot ; {\cal H},{\cal H}_0) \in
L^1(\re; \langle E \rangle^{-\gamma - 1} dE)$ which vanishes
identically on $(-\infty, E_0)$ such that {\em the Lifshits-Krein
formula}
    \bel{10LK}
{\rm Tr} (f({\cal H}) - f({\cal H}_0)) = \int_{\re} \xi(E; {\cal
H}, {\cal H}_0) f'(E) dE
    \ee
holds for each $f \in C_0^{\infty}(\re)$ (see \cite{lif},
\cite{kr}, \cite[Chapter 8]{Y2}). The function $\xi(.; {\cal H}, {\cal H}_0)$ is called
the spectral shift function (SSF) for the pair of the operators $({\cal H}, {\cal H}_0)$. If
$E <  \inf \sigma({\cal H}_0)$, then the spectrum of ${\cal H}$
below $E$ could be at most discrete, and for almost every $E <
\inf \sigma({\mathcal H}_0)$ we have
\begin{equation} \label{101a}
\xi(E; {\cal H}, {\cal H}_0) = - N(E; {\cal H})
\end{equation}
where   $ N(E; {\cal H}) :=   {\rm rank}\,{\mathbb
P}_{(-\infty,E)}({\cal H})$. On the other hand,  for almost every
$E \in \sigma_{\rm ac}({\cal H}_0)$, the SSF $\xi(E; {\cal H},
{\cal H}_0)$ is related to the scattering determinant ${\rm
det}\,S(E; {\cal H}, {\cal H}_0)$ for the pair $({\cal H}, {\cal
H}_0)$ by {\em the Birman-Krein} formula
    \bel{rr9}
{\rm det}\;S(E; {\cal H}, {\cal H}_0) = e^{-2\pi i \xi(E; {\cal
H}, {\cal H}_0)}
    \ee
(see \cite{BK}, \cite[Section 8.4]{Y2}). \\
Next, we define the SSF for the pair $(H,H_0)$. We will say that
$V$ satisfies condition ${\mathcal D}_{\alpha}$, $\alpha \in \re$,
if
$$
|V(x,y)| \leq c \langle y \rangle^{-\alpha}, \; c>0, \; (x,y) \in
S_L,
$$
where, as usual, $\langle y \rangle : = (1 + y^2)^{1/2}$.
 Assume
that $V$ satisfies condition ${\mathcal D}_{\alpha}$ with $\alpha
> 1$. Then \eqref{10} holds for ${\mathcal H} = H$, ${\mathcal
H}_0 = H_0$, and $\gamma = 1$, and hence the SSF $\xi(\cdot ;
H,H_0)$ is well defined as an element of $L^1(\re; \langle E
\rangle^{-2}dE)$. In the present article we will identify this SSF
with a representative of the corresponding class of equivalence
described explicitly in Section 4.3 below.
\begin{pr} \label{ap21}
Assume that $V$ satisfies ${\mathcal D}_{\alpha}$ with $\alpha >
1$. Then the SSF $\xi(\cdot;H,H_0)$ is bounded on every compact
subset of $\re\setminus{\mathcal Z}$ and continuous on
$\re\setminus\left({\mathcal Z} \cup \sigma_p(H)\right)$.
\end{pr}
The proof of Proposition \ref{ap21} can be found in Subsection 4.6
below.\\
     Set
$$
J(x,y) = {\rm sign}\,V(x,y): = \left\{
\begin{array} {l}
1 \quad {\rm if} \quad V(x,y) \geq 0,\\
-1 \quad {\rm if} \quad V(x,y) < 0.
\end{array}
\right.
$$
Fix $j \in {\mathbb N}$. Let $\psi_j(\cdot;k) : I_L \to \re$, $k
\in \re$, be the real-valued normalized in $L^2(I_L)$
eigenfunction of the operator $\hat{H}(k)$ corresponding to the
eigenvalue $E_j(k)$. For $\varepsilon \in (-1,1)$ introduce {\em
the effective potential}
$$
w_{j, \varepsilon}(y) : = \int_{I_L} |V(x,y)|(J(x,y) -
\varepsilon)^{-1} \psi_j(x;0)^2 dx, \quad y \in \re,
$$
so that $w_{j, 0}(y)  = \int_{I_L} V(x,y) \psi_j(x;0)^2 dx$, and
{\em the effective Hamiltonians}
$$
h_{0,j} : = - \mu_j \frac{d^2}{dy^2}, \quad h_{j}(\varepsilon) : =
h_{0,j} + w_{j,\varepsilon},
$$
the number $\mu_j$ being defined in \eqref{F1}.  Note if $V$
satisfies ${\mathcal D}_{\alpha}$ with $\alpha>1$, then \eqref{10}
holds for ${\mathcal H} = h_{j}(\varepsilon)$, ${\mathcal H}_0 =
h_{0,j}$, and $\gamma = 1$, and hence the SSFs $\xi(\cdot ;
h_{j}(\varepsilon), h_{0,j})$, $j \in {\mathbb N}$,
$\varepsilon \in (-1,1)$, are well defined.\\
For $\lambda > 0$ set
    \bel{m7}
\theta_{\beta}(\lambda) : = \left\{
\begin{array} {l}
1 \quad {\rm if} \quad \beta > 1/2, \\
|\ln{\lambda}| \quad {\rm if} \quad \beta = 1/2,\\
\lambda^{-\frac{1}{2} + \beta} \quad {\rm if} \quad 0 < \beta <
1/2.
\end{array}
\right.
    \ee
If $\lambda < 0$, then
    \bel{m8}
\theta_{\beta}(\lambda) : = 1
    \ee
for all $\beta > 0$.
\begin{theorem} \label{t2} Assume that $V$ satisfies
${\mathcal D}_{\alpha}$ with $\alpha > 1$. Fix $q \in {\mathbb
N}$. Then for each $\varepsilon \in (0,1)$ we have
    \bel{m18}
\xi(\lambda; h_q(-\varepsilon), h_{0,q}) +
O(\theta_{2\gamma}(\lambda)) \leq \xi({\mathcal E}_q + \lambda;
H,H_0) \leq \xi(\lambda; h_q(\varepsilon), h_{0,q}) +
O(\theta_{2\gamma}(\lambda)),
    \ee
as $\lambda \to 0$, for any $\gamma \in \left(0,
(\alpha-1)/2\right)$, $\gamma \leq 1$.
\end{theorem}
The proof of Theorem \ref{t2} can be found in Subsection 4.7.\\
Assume now that $\alpha \in (1,2)$. Then there exists $\gamma \in
\left(0, (\alpha-1)/2\right)$, $\gamma \leq 1$, such that
$\theta_{2\gamma}(\lambda) =
o(|\lambda|^{\frac{1}{2}-\frac{1}{\alpha}})$ as $\lambda \to 0$.
Hence, using well-known results concerning the asymptotic
behaviour of the SSF $\xi(\lambda ; h_{j}(\varepsilon), h_{0,j})$
as $\lambda \to 0$ (see e.g. \cite[Theorem XIII.82]{RS} in the
case $\lambda \uparrow 0$, and \cite{Y1} in the case $\lambda
\downarrow 0$), we obtain the following
\begin{follow} \label{f1}
    Let $V$ satisfy ${\mathcal D}_{\alpha}$ with
$\alpha \in (1,2)$. Fix $q \in {\mathbb N}$. Suppose that for each
$\varepsilon \in (-\varepsilon_0, \varepsilon_0)$ and some
$\varepsilon_0 \in (0,1)$ there exist real numbers
$\omega_{q,\pm}(\varepsilon)$ such that
    \bel{sof1}
    \lim_{y \to \pm
\infty} |y|^{\alpha} w_{q,\varepsilon}(y) =
\omega_{q,\pm}(\varepsilon)
    \ee
uniformly with respect to $\varepsilon$. Then we have
    \bel{a6}
\lim_{\lambda \downarrow 0} \lambda^{\frac{1}{\alpha} -
\frac{1}{2}} \xi({\mathcal E}_q - \lambda;H,H_0) = -\mu_q^{-1/2}
{\mathcal C}_{\alpha} \Omega_q^-,
    \ee
    \bel{rr3}
\lim_{\lambda \downarrow 0} \lambda^{\frac{1}{\alpha} -
\frac{1}{2}} \xi({\mathcal E}_q + \lambda;H,H_0) = -\mu_q^{-1/2}
{\mathcal C}_{\alpha} \left(\csc{(\pi/\alpha)}\Omega_q^- +
\cot{(\pi/\alpha)}\Omega_q^+\right),
    \ee
 where ${\mathcal C}_{\alpha}: = \frac{1}{\pi} \int_0^1
(t^{-\alpha} - 1)^{1/2} dt$, and $\Omega_q^{\pm} : =
\sum_{\varsigma = +,-} \omega_{q,\varsigma}(0)_{\pm}^{1/\alpha}$,
while $\omega_{q,\varsigma}(0)_{+}$ and
$\omega_{q,\varsigma}(0)_{-}$ denote the positive and the negative
part of $\omega_{q,\varsigma}(0)$ respectively.
\end{follow}
For the sake of completeness  we include a
sketch of the proof of Corollary \ref{f1} in Subsection 4.8. \\
 {\em Remark}: If $q=1$ and $\lambda>0$,  we have
    $\xi({\mathcal E}_1-\lambda; H,H_0) = - N({\mathcal
    E}_1-\lambda; H)$ (cf. \eqref{101a}). Note that the spectrum of $H$ below
${\mathcal E}_1$ is discrete if $V$ satisfies ${\mathcal
D}_{\alpha}$ with any $\alpha > 0$. Moreover,   as in \eqref{a6} we have
    \bel{a7}
    \lim_{\lambda \downarrow 0} \lambda^{\frac{1}{\alpha} -
\frac{1}{2}}N({\mathcal
    E}_1-\lambda; H)  = \mu_q^{-1/2} {\mathcal C}_{\alpha}
    \Omega_1^-
    \ee
      for all $\alpha \in (0,2)$ and not only for $\alpha \in (1,2)$. \\
Similarly, using well known results on the asymptotic behaviour as
$\lambda \uparrow 0$ of the SSF $\xi(\lambda ; h_{q}(\varepsilon),
h_{0,q})$  in the case $\alpha = 2$ (see \cite{KS}), we obtain the
following
\begin{follow} \label{f2}
Assume the hypotheses of Corollary \ref{f1} with $\alpha = 2$. Fix
$q \in {\mathbb N}$. Then we have
$$
\lim_{\lambda \downarrow 0} |\ln{\lambda}|^{-1} \xi({\mathcal E}_q
- \lambda;H,H_0) = - \frac{1}{2\pi} \sum_{\varsigma = +,-}
\left(\frac{\omega_{q,\varsigma}(0)}{\mu_q} +
\frac{1}{4}\right)_-^{1/2}.
$$
Moreover, if $\omega_{q,\pm}(0) > -\mu_q/4$, then $\xi({\mathcal
E}_q - \lambda;H,H_0) = O(1)$ as $\lambda \downarrow 0$.
\end{follow}
{\em Remark}: In the case $\alpha = 2$, the analysis of the
asymptotic behaviour of $\xi(\lambda ; h_{j}(\varepsilon),
h_{0,j})$ as $\lambda \downarrow 0$ requires some additional
estimates similar to those obtained in \cite{Y1}. In order to
avoid the inadequate increase
of the size of the article, we omit these results. \\
Finally, in Subsection 4.9 we prove
\begin{follow} \label{f3}
Let $V$ satisfy ${\mathcal D}_{\alpha}$ with $\alpha > 2$. Then
for each $q \in {\mathbb N}$ we have
    \bel{m31}
\xi({\mathcal E}_q + \lambda;H,H_0) = O(1), \quad \lambda \to 0.
    \ee
\end{follow}
\section{Mourre estimates}
\setcounter{equation}{0}
    In this section we prove Theorem \ref{t1}   using an
    appropriate Mourre estimate established in Proposition
    \ref{Moest2}. Similar Mourre estimates have been obtained in \cite{dBP}
    for a
    2D magnetic Schr\"odinger operator defined on the half-plane, and
    in \cite[Chapter 3]{GL} for a
    3D one defined in the whole space.
\begin{lemma}
\label{lm-sep}  Let $n \in {\mathbb N}$,   $E \in
(\mathcal{E}_{n}, \mathcal{E}_{n+1})$. Then there exists $\delta =
\delta(E) \in (0, {\rm dist}\,(E,{\mathcal Z}))$ such that the
interval $[E-\delta ,E+\delta]$ satisfies
    \bel{m22}
E_r^{-1}([E-\delta ,E+\delta]) = \emptyset, \quad r \geq n+1,
    \ee
and, if $n \geq 2$,
    \bel{eqsep1}
    E_r^{-1}([E-\delta ,E+\delta]) \cap E_s^{-1}([E-\delta ,E+\delta]) = \emptyset, \quad
    r \neq s, \quad r,s = 1,\ldots,n.
    \ee
\end{lemma}
\begin{proof}
First, \eqref{m22} follows trivially from $[E-\delta ,E+\delta] \cap
[{\mathcal
E}_{n+1}, \infty) = \emptyset$.\\
Set $B_r : = E_r^{-1}([E-\delta ,E+\delta]) \cap [0,\infty)$, $r=1,\ldots,n$.
Since $E_r$ are even functions of $k$, it suffices to show that
    \bel{m20}
    B_r \cap B_s = \emptyset, \quad
    r \neq s, \quad r,s = 1,\ldots,n,
    \ee
      instead of \eqref{eqsep1}.
Denote by $E_r^{-1}$, $r \in {\mathbb N}$, the function inverse to
$E_r : [0,\infty) \to \re$. Since $[E-\delta ,E+\delta] \subset ({\mathcal
E}_n, \infty)$, this interval is in the domain of all the
functions $E_r^{-1}$, $r = 1, \ldots, n$, and we have
    $$
    B_r = [E_r^{-1}(E-\delta), E_r^{-1}(E+\delta)], \quad r = 1, \ldots, n.
    $$
Therefore, in order to prove that there exists $\delta \in (0,
{\rm dist}\,(E,{\mathcal Z}))$ such that \eqref{m20} holds true,
it suffices to show that there exists $\delta \in (0, {\rm
dist}\,(E,{\mathcal Z}))$ such that
$$
E_{r+1}^{-1}(E+\delta) < E_r^{-1}(E-\delta), \quad r=1, \ldots,
n-1,
$$
 which is evident since $E_{r+1}^{-1}(E) < E_r^{-1}(E)$,  the
 functions $E_r^{-1}$ are continuous, and $n-1$ is finite.
\end{proof}
\begin{lemma}
\label{diffpro} Assume \eqref{f4a}. Let   $\chi \in C_0^\infty
(\re)$. Then $\chi(H)-\chi(H_0) \in S_{\infty}$.
\end{lemma}
\begin{proof}
By the Helffer-Sj\"ostrand formula, we have
$$
  \chi(H)-\chi(H_0)= - \frac{1}{\pi} \int_{\re^2} \frac{\partial
    \tilde{\chi}}{\partial \bar{z}} (H-z)^{-1} V (H_0 - z)^{-1}
    dx dy
$$
where $z = x+iy$, $\bar{z} = x - iy$, $\tilde{\chi}$ is the
quasi-analytic extension of $\chi$, and the convergence of the the
integral is understood in the operator-norm sense (see e.g.
\cite[Chapter 8]{dsj}). Since the support of $\tilde{\chi}$ is
compact in $\re^2$, and the operator $\frac{\partial
    \tilde{\chi}}{\partial \bar{z}} (H-z)^{-1} V (H_0 - z)^{-1}$ is compact for
every $(x,y) \in \re^2$ with $y \neq 0$, and is uniformly
norm-bounded   on $\re^2$, we have $\chi(H)-\chi(H_0) \in
S_{\infty}$.
\end{proof}
 Introduce the operator
$$
A = A^* = -\frac{i}{2} \left( y \frac{\partial}{\partial y} +
\frac{\partial}{\partial  y} y \right)
$$
defined originally on $C_0^{\infty}(\re_y; D(\hat{H}))$ and then
closed in $L^2({\mathcal S}_L)$. Note that
$$
(e^{itA} f) (x,y) =  e^{t/2} f(x, e^t y), \quad t \in \re, \quad f
\in L^2({\mathcal S}_L),
$$
 and the unitary group $e^{itA}$ preserves $D(H_0)$. In what follows,
 we will consider   $D(H_0^{\gamma})$, $\gamma > 0$, as a Hilbert space equipped with the scalar product
$\langle H_0^{\gamma} u, H_0^{\gamma} v \rangle_{L^2({\mathcal
S}_L)}$, $u, v \in D(H_0^{\gamma})$. Denote by
$D(H_0^{\gamma})^*$, $\gamma > 0$, the completion of
$L^2({\mathcal S}_L)$ with respect to the norm
$\|H_0^{-\gamma}u\|_{L^2({\mathcal S}_L)}$, $u \in L^2({\mathcal
S}_L)$.\\
 Note that $C_0^{\infty}(\re_y; D(\hat{H}))$ is dense in
$D(H_0)$, and, hence, $D(A) \cap D(H_0)$ is dense in $D(H_0)$.
\begin{pr}
\label{Moest2}  Assume \eqref{f4a} -- \eqref{f4}. Let $n \in
{\mathbb N}$, $E \in (\mathcal{E}_n,\mathcal{E}_{n+1})$. Assume
that $\delta \in (0, {\rm dist}(E, {\mathcal Z}))$ is chosen to
satisfy \eqref{m22} and \eqref{eqsep1} according to Lemma
\ref{lm-sep}. Let $\chi \in C_0^\infty (\re)$, ${\rm supp}\,\chi =
[E-\delta, E+\delta]$. Then there exists $K \in S_{\infty}$ and a
constant $C>0$ such that
    \bel{252}
    \chi(H) [H,iA] \chi(H) \geq C \chi(H)^2 + K
    \ee
where the commutator $[H,iA]$ is understood as a bounded operator
from $D(H_0)$ into $D(H_0)^*$.
\end{pr}
\begin{proof}
 A straightforward calculation yields
    \bel{f10a}
 [H,iA] = [H_0,iA] + [V,iA]
  \ee
 where
    \bel{comm}
    [H_0,iA] = -2 \frac{\partial^2}{\partial y^2} +
2 i b x \frac{\partial}{\partial y},
    \ee
and
    \bel{26}
[V,iA]=-y \frac{\partial V(x,y)}{\partial y}.
    \ee
Evidently, $[H_0,iA]$ is a bounded operator from $D(H_0)$ into
$L^2({\mathcal S}_L)$, and, hence, is a bounded operator from
$D(H_0)$ into $D(H_0)^*$. On the other hand, $[V,iA]$ is a compact
operator from $D(H_0)$ into $D(H_0)^*$. Hence, $[H,iA]$ is a
bounded operator from $D(H_0)$ into $D(H_0)^*$.  Further, for
$\chi \in C_0^{\infty}(\re)$ we have
 \bel{23}
\chi(H_0) [H_0,iA] \chi(H_0)=
 \mathcal{F}^* \left(
2\sum_{r,s=1}^{\infty} \int_{\re}^{\oplus} \chi(E_r(k))
\chi(E_s(k)) k p_r(k) (k-bx) p_s(k) dk \right)\mathcal{F},
    \ee
    where
    \bel{m32}
    p_r(k) : = \langle \cdot, \psi_r(\cdot;k)\rangle \psi_r(\cdot;k),
\quad k \in \re, \quad r \in {\mathbb N},
    \ee
    $\psi_r(\cdot;k)$ being the eigenfunction defined in
    Subsection 2.4 before the formulation of Theorem \ref{t2}.
Using \eqref{m22} and \eqref{eqsep1}, we find that  (\ref{23})
reduces to
    \bel{239}
    \chi(H_0) [H_0,iA] \chi(H_0)  =
 2\mathcal{F}^* \left(\sum_{r=1}^{n} \int_{\re}^{\oplus}
 \chi(E_r(k))^2k
 \langle (k-bx) \psi_r(k) , \psi_r(k) \rangle p_r(k) dk \right)\mathcal{F}.
\ee This, combined with the Feynman-Hellmann formula
    \bel{f10}
E_r'(k) = 2 \langle (k-bx) \psi_r(k) , \psi_r(k) \rangle,
    \ee
yields
    \bel{24}
    \chi(H_0) [H_0,iA] \chi(H_0) =
\mathcal{F}^* \left( \sum_{r=1}^{n}  \int_{\re}^{\oplus} k E_r'(k)
\chi(E_r(k))^2 p_r(k) dk\right) \mathcal{F}.
    \ee
    Moreover, by
\eqref{m21}, we have
$$
k E_r'(k) \chi(E_r(k))^2  \geq C_r \chi(E_r(k))^2,
$$
with $C_r = \min_{k \in [E-\delta, E+\delta]} kE'_r(k) > 0$, $r =
1, \ldots, n$. Therefore,
 \bel{241}
 \chi(H_0) [H_0,iA] \chi(H_0)
\geq C  \mathcal{F}^*\left( \sum_{r=1}^{n} \int_{\re}^{\oplus}
\chi(E_r(k))^2 p_r(k)\right)  \mathcal{F} =C \chi(H_0) ^2, \ee
 where $C:= \min_{r=1,\ldots,n} C_r > 0$.
 By \eqref{f10a},
    \bel{27}
\chi(H) [H,i A] \chi(H) = \chi(H_0) [H_0,i A] \chi(H_0) + K_0,
    \ee
    where
$$
    K_0  = \chi(H_0) [H_0,i A]
\left( \chi(H)-\chi(H_0) \right) + \left( \chi(H)-\chi(H_0)
\right) [H_0,i A] \chi(H) - \chi(H)y \frac{\partial V}{\partial y}
\chi(H): =
$$
    $$
    K_1 + K_2 + K_3.
    $$
  We have
  $$
  K_1  =  \chi(H_0) H_0 H_0^{-1}[ H_0, iA ]\left( \chi(H)-\chi(H_0) \right),
  $$ and the operators
$\chi(H_0) H_0$ and $H_0^{-1}[ H_0, iA ]$ extend to bounded
operators in $L^2({\mathcal S}_L)$ (see \eqref{comm}). Since the
operator $\chi(H)-\chi(H_0)$ is compact by Lemma \ref{diffpro}, we
conclude that $K_1 \in S_{\infty}(L^2({\mathcal S}_L))$.
Similarly, taking into account that $\chi(H)-\chi(H_0)$ is
compact, and the operators $[ H_0,i A] H_0^{-1}$ and $H_0\chi(H) =
H\chi(H) - V \chi(H)$ are bounded, we get
 $$
 K_2=\left(
\chi(H)-\chi(H_0) \right) [ H_0,i A] H_0^{-1} H_0 \chi(H) \in
S_{\infty}(L^2({\mathcal S}_L)).
$$
 Finally, the  operator
$$
 K_3= \chi(H)  y
\frac{\partial V}{\partial y} \chi(H)=
 \chi(H)H_0 H_0^{-1} y \frac{\partial V}{\partial y} H_0^{-1} H_0 \chi(H)
 $$ is compact in $L^2({\mathcal S}_L)$ since $H_0^{-1} y \frac{\partial V}{\partial y} H_0^{-1}$
 is compact by \eqref{f4}, and  $\chi(H)H_0 = (H_0\chi(H))^*$ is
 bounded in $L^2({\mathcal S}_L)$.
Therefore, $K_0 = K_1+ K_2+K_3 \in S_{\infty}$. Combining
\eqref{241} and \eqref{27}, we get
    \bel{f20}
\chi(H) [H,i A] \chi(H) \geq C \chi(H_0)^2 + K_0 = C \chi(H)^2 +
K_0 + K_4,
    \ee
where $K_4 : = C\left(\chi(H_0)^2 - \chi(H)^2\right) \in
S_{\infty}$ by Lemma \ref{diffpro}. Hence \eqref{f20} implies
\eqref{252} with $K = K_0 + K_4$.
\end{proof}
For $E \in \re$ and $\delta > 0$ set $\Delta_E(\delta) : =
(E-\delta/2,E+\delta/2)$.
\begin{follow}
\label{Moest3} Assume \eqref{f4a} -- \eqref{f4}. Fix  $E \in
(\mathcal{E}_n,\mathcal{E}_{n+1})$, $n \in {\mathbb N}$. Let
$\delta \in (0, {\rm dist}\,(E,{\mathcal Z}))$  be chosen as in
Proposition
\ref{Moest2}. \\
(i) We have
    \bel{253}
    {\mathbb P}_{\Delta_E(\delta)}(H)[H,iA] {\mathbb
    P}_{\Delta_E(\delta)}(H)
\geq C {\mathbb
    P}_{\Delta_E(\delta)}(H) + \tilde{K}
    \ee
where $\tilde{K} : = {\mathbb P}_{\Delta_E(\delta)}(H) K {\mathbb
P}_{\Delta_E(\delta)}(H) \in S_{\infty}$, $C$ and $K$ being the
same as in
\eqref{252}.\\
(ii) Suppose moreover that $E \not \in \overline{\sigma_p(H)}$.
Then for $\delta' \in (0,\delta)$ small enough we have
    \bel{253a}
{\mathbb P}_{\Delta_E(\delta')}(H)[H,iA] {\mathbb
    P}_{\Delta_E(\delta')}(H)
\geq \frac{1}{2} C {\mathbb
    P}_{\Delta_E(\delta')}(H).
    \ee
\end{follow}
\begin{proof}
Choose $\chi$ in $\eqref{252}$ to be equal to one on
$\Delta_E(\delta)$, and multiply \eqref{252} from the left and the
right by ${\mathbb P}_{\Delta_E(\delta)}(H)$. Thus we get
\eqref{253}. In order to obtain \eqref{253a}, we repeat the
argument of the proof of \cite[Lemma 4.8]{CFKS}. Pick $\delta' \in
(0,\delta)$ and multiply \eqref{253} from the right and the left
by ${\mathbb P}_{\Delta_E(\delta')}(H)$. We get
    \bel{f80}
{\mathbb P}_{\Delta_E(\delta')}(H)[H,iA] {\mathbb
    P}_{\Delta_E(\delta')}(H)
\geq C {\mathbb
    P}_{\Delta_E(\delta')}(H) + {\mathbb P}_{\Delta_E(\delta')}(H)\tilde{K}{\mathbb
    P}_{\Delta_E(\delta')}(H).
    \ee
Since $E \not \in \overline{\sigma_p(H)}$  and, hence, ${\rm
s}-\lim_{\delta' \downarrow 0}{\mathbb P}_{\Delta_E(\delta')}(H) =
0$, while $\tilde{K}$ is compact, we have ${\rm n}-\lim_{\delta'
\downarrow 0}{\mathbb P}_{\Delta_E(\delta')}(H)\tilde{K}{\mathbb
    P}_{\Delta_E(\delta')}(H) = 0$. Choose $\delta' \in (0,\delta)$
    so small that $$\|{\mathbb
P}_{\Delta_E(\delta')}(H)\tilde{K}{\mathbb
    P}_{\Delta_E(\delta')}(H)\| \leq C/2$$ which implies
    \bel{f81}
    {\mathbb
P}_{\Delta_E(\delta')}(H)\tilde{K}{\mathbb
    P}_{\Delta_E(\delta')}(H) \geq -\frac{1}{2} C {\mathbb
P}_{\Delta_E(\delta')}(H).
    \ee
Combining \eqref{f80} with \eqref{f81}, we obtain \eqref{253a}.
\end{proof}
Since the unitary group $e^{itA}$ preserves $D(H_0)$, and $[H,iA]
: D(H_0) \to D(H_0)^*$ is a bounded operator, the Mourre estimate
\eqref{253}  entails the following
\begin{follow} \label{valp} {\rm \cite{Mo}, \cite[Theorem 4.7]{CFKS}, \cite{GG}}
Assume \eqref{f4a} -- \eqref{f4}. Let $E$, $\delta$, and
$\Delta_E(\delta)$ be as in Corollary \ref{Moest3}. Then
$\Delta_E(\delta)$ contains at most finitely  many eigenvalues of
$H$, each of them having a finite multiplicity.
\end{follow}
Now we are in position to prove Theorem \ref{t1}. Let $\Delta
\subset \re \setminus {\mathcal Z}$ be a compact interval. If
$\Delta \subset (-\infty, {\mathcal E}_1)$, then $\Delta \cap
\sigma_{\rm ess}(H) = \emptyset$ and $\Delta$ may contain at most
a finite number of eigenvalues, each having a finite multiplicity.
Assume $\Delta \subset ({\mathcal E}_n, {\mathcal E}_{n+1})$, $n
\in {\mathbb N}$. For each $E \in \Delta$ choose $\delta =
\delta(E)$ as in Proposition \ref{Moest2}. Then we have $\Delta
\subset \cup_{E \in \Delta}   \Delta_E(\delta)$. Since $\Delta$ is
compact, there exists a finite set $\{E_j\}_{j=1}^{N}$ of energies
$E_j \in \Delta$ such that
    \bel{f6}
    \Delta \subset \cup_{j=1}^N   \Delta_{E_j}(\delta).
    \ee
Assume \eqref{f4a} -- \eqref{f4}. Then \eqref{f6} and  Corollary
\ref{valp} imply that $\Delta$ may contain at most a finite number
of eigenvalues, each having a finite multiplicity. Hence, the
first part of
Theorem \ref{t1} is proved. \\
Assume moreover \eqref{f5a} -- \eqref{f5}.   It follows from
\eqref{f5a}  that $[H,iA]$ extends to a bounded operator from
$D(H_0)$ to $D(H_0^{1/2})^*$, while \eqref{f5} combined with
\eqref{f4}, implies that the second commutator $[[H,iA],iA]$
extends to a bounded operator from $D(H_0)$ to $D(H_0)^*$. Then
Corollary 3.1 ii) together with the results of \cite[Corollary
4.10]{CFKS} and \cite{GG} (see also \cite{Mo}) imply that
$\sigma_{\rm sc}(H) \cap \left(({\mathcal E}_n, {\mathcal
E}_{n+1}) \setminus \overline{\sigma_p(H)}\right) = \emptyset$, $n
\in {\mathbb N}$. Since the set $({\mathcal E}_n, {\mathcal
E}_{n+1}) \cap \overline{\sigma_p(H)}$ is at most discrete, we get
$\sigma_{\rm sc}(H) \cap ({\mathcal E}_n, {\mathcal E}_{n+1}) =
\emptyset$, $n \in {\mathbb N}$. Finally, since ${\mathcal E}_1 =
\inf \sigma_{\rm ess}(H)$ we have $\sigma_{\rm sc}(H) \cap
(-\infty, {\mathcal E}_1) = \emptyset$. Therefore, $\sigma_{\rm
sc}(H) \cap (\re \setminus {\mathcal Z}) = \emptyset$. Since
${\mathcal Z}$ is discrete, $\sigma_{\rm sc}(H) = \emptyset$. The
second part of Theorem \ref{t1} is now proved too. \\
{\em Remark}:   Mourre estimates and their corollaries concerning
the spectrum of $H$ could be also deduced from the general scheme
for analytically fibered operators developed in \cite{gn}. The
advantage of our approach is that it relies on an explicit and
simple conjugate operator $A$, and offers an explicit description
of the ``exceptional set'' ${\mathcal Z}$.

\section{Analysis of the Spectral Shift Function}
\setcounter{equation}{0}
    {\bf 4.1.}  In this subsection we summarize some
    simple properties of compact operators which will be systematically used in the sequel. For
$s>0$ and $T^* = T  \in S_{\infty}$ set
$$
n_{\pm}(s;T): = {\rm rank}\; {\mathbb P}_{(s,\infty)}(\pm T).
$$
For an arbitrary (not necessarily self-adjoint) operator $T\in
S_{\infty}$ put
\begin{equation} \label{31}
n_*(s;T): = n_+(s^2;T^*T), \quad s>0.
\end{equation}
If $T = T^*$, then evidently
\begin{equation} \label{num}
  n_*(s;T) = n_+(s,T) + n_-(s;T), \quad s>0.
\end{equation}
If $T_1, T_2 \in S_{\infty}$, and $s_1>0$, $s_2>0$, then the well
known Weyl -- Ky Fan inequalities
    \bel{a19}
    n_*(s_1 + s_2; T_1 + T_2) \leq n_*(s_1; T_1) + n_*(s_2; T_2)
    \ee
    hold true. Moreover, if $T_j = T_j^*$, $T_1 \in S_{\infty}$, and ${\rm rank}\,T_2 < \infty$, we have
\bel{m15}
 n_{\pm}(s; T_1) - {\rm rank}\,T_2 \leq   n_{\pm}(s; T_1 + T_2)
 \leq n_{\pm}(s; T_1) + {\rm rank}\,T_2, \quad s>0.
    \ee
 If $T \in S_p$, $p \in
[1,\infty)$, then the following elementary Chebyshev-type
inequality
\begin{equation} \label{32a}
n_*(s;T) \leq s^{-p} \|T\|_p^p
\end{equation}
holds for every $s>0$.\\

{\bf 4.2.} In this subsection we introduce the concepts of index
of a Fredholm pair of orthogonal projections, and index for a pair
of selfadjoint operators, and discuss some of their properties.
More details can be found in \cite{ASS} and \cite{BPR}.\\
A pair of orthogonal projections $(P, Q)$ is said to be Fredholm
if
$$
\left\{-1,1\right\} \cap \sigma_{\rm ess}(P-Q) = \emptyset.
$$
In particular, if $P-Q \in S_{\infty}$, then the pair $(P,Q)$ is Fredholm. \\
Assume that the pair of orthogonal projections $(P,Q)$ is
Fredholm. Set
$$
{\rm index}(P,Q)  : = {\rm dim \; Ker}\;(P-Q-I) - {\rm dim \;
Ker}\;(P-Q+I).
$$
Let $\tilde{M}$, $M$, be bounded self-adjoint operators. If the
spectral projections ${\mathbb P}_{(-\infty,0)}(\tilde{M})$ and
${\mathbb P}_{(-\infty,0)}(M)$ form a Fredholm pair, we will use
the notation
$$
{\rm ind}(\tilde{M}, M)  : = {\rm index}({\mathbb
P}_{(-\infty,0)}(\tilde{M}),  {\mathbb P}_{(-\infty,0)}(M)).
$$
A sufficient condition that  the pair ${\mathbb
P}_{(-\infty,0)}(\tilde{M}), {\mathbb P}_{(-\infty,0)}(M)$ be
Fredholm, is $\tilde{M} = M + A$ where $M$ is a bounded
self-adjoint operator such that $0 \not \in \sigma_{\rm ess}(M)$,
and $A = A^* \in S_{\infty}$.
 \begin{lemma} \label{lparis2} {\rm  \cite[Subsection 3.2]{BPR}}
 Let $M$ be a bounded self-adjoint operator such that
$0 \not \in \sigma(M)$. Let $A$ and $B$ be compact self-adjoint
operators. Then for $s \in (0,\infty)$ such that $[-s,s] \cap
\sigma(M) = \emptyset$ we have
    \bel{paris2} {\rm ind}(M+s+B, M+s)
- n_+(s;A) \leq {\rm ind}(M+A+B, M) \leq {\rm ind}(M-s+B, M-s) +
n_-(s;A).
    \ee
    Assume, moreover, that the rank of $A$ is finite.
    Then we have
    \bel{a20a}
    {\rm ind}(M+B, M)
- {\rm rank}\, A \leq {\rm ind}(M+A+B, M) \leq {\rm ind}(M+B, M) +
{\rm rank}\, A.
    \ee
\end{lemma}
{\em Remark}: Note that in the case $B=0$, estimates
\eqref{paris2} imply
    \bel{a20}
|{\rm ind}(M+A, M)| \leq n_*(s;A)
    \ee
    for any $s>0$ such that $[-s,s] \cap
\sigma(M) = \emptyset$.
\begin{lemma} \label{l31} {\rm \cite[Lemma 2.1]{P1}, \cite[Subsection 3.3]{BPR}}
Let $M$ be a bounded self-adjoint operator such that
$0 \not \in \sigma(M)$. Let $T_1 = T_1^* \in S_{\infty}$ and $T_2
= T_2^* \in S_1$. Then for each $s_1>0$, $s_2 >0$ such that
$[-s,s] \cap \sigma(M) = \emptyset$ with $s=s_1+s_2$,  we have
\begin{equation} \label{pu7}
\int_{\re} \left|{\rm ind}(M + T_1 + t \, T_2, M)\right| \,d\mu(t)
\leq n_{*}(s_1;T_1) + \frac{1}{\pi s_2}\|T_2\|_1
\end{equation}
where $d\mu(t): = \frac{1}{\pi} \frac{dt}{1+t^2}$.
\end{lemma}
{\bf 4.3.} In this subsection we describe a representation of the
SSF $\xi(E;{\mathcal H},{\mathcal H}_0)$ which is a special case
of the general representation of the SSF due to
F. Gesztesy, K. Makarov, and A. Pushnitski (see \cite{P1}, \cite{GM}, \cite{P2}). \\
 Let $X_1$
and $X_2$ be two Hilbert spaces. Let ${\mathcal H}$ and ${\mathcal
H}_0$ be two lower bounded self-adjoint operators acting in $X_1$.
Assume that \eqref{10} holds for some $\gamma > 0$. Next suppose
that
    \bel{sof2}
    {\mathcal V} : = {\mathcal H} - {\mathcal H}_0 = {\mathcal
    K}^* {\mathcal J} {\mathcal K}
    \ee
where ${\mathcal K} \in {\mathcal B}(X_1, X_2)$,  ${\mathcal J} =
{\mathcal J}^* \in {\mathcal B}(X_2)$, and $0 \not \in
\sigma({\mathcal J})$. Finally, assume that
    \bel{sof3}
    {\mathcal K} ({\mathcal H}_0 - E_0)^{-1/2} \in
    S_{\infty}(X_1, X_2),
    \ee
    \bel{sof4}
    {\mathcal K} ({\mathcal H}_0 - E_0)^{-\gamma'} \in
    S_2 (X_1, X_2),
    \ee
for some $E_0 < \inf{\sigma({\mathcal H}) \cup \sigma({\mathcal
H}_0)}$ and $\gamma'>0$. For $z \in {\mathbb C}_+ : = \{\zeta \in
{\mathbb C}\,|\, {\rm Im}\,\zeta > 0\}$ set
$$
{\mathcal T}(z): = {\mathcal K} ({\mathcal H}_{0} - z)^{-1}
{\mathcal K}^*.
$$
Evidently, ${\mathcal T}(z) \in S_{\infty}(X_2)$.
\begin{lemma} \label{sofl1} {\rm \cite{BE}} Let
\eqref{sof2} -- \eqref{sof4} hold true. Then for almost every  $E
\in \re$ the operator-norm  limit ${\mathcal T}(E): = {\rm
n}-\lim_{\delta \downarrow 0} {\mathcal T}(E + i\delta)$ exists
and by \eqref{sof4} we have ${\mathcal T}(E) \in S_{\infty}(X_2)$.
Moreover, $0 \leq {\rm Im}\,{\mathcal T}(E) \in S_1(X_2)$.
\end{lemma}
\begin{theorem} \label{soft1}{\rm \cite{P1}, \cite{GM}, \cite{P2}} Let
\eqref{10} and \eqref{sof2} -- \eqref{sof4} hold true. Then for
almost every $E \in \re$ we have
    \bel{sof5a}
    \xi(E;{\mathcal H}, {\mathcal H}_0) =
    \int_{\re} {\rm ind}({\mathcal J}^{-1} +  {\rm Re} \,
{\mathcal T}(E) + t\; {\rm Im} \, {\mathcal T}(E), {\mathcal
J}^{-1})\; d\mu (t).
    \ee
\end{theorem}
Note that the convergence of the integral in
\eqref{sof5a} is guaranteed by Lemma \ref{l31}.\\
Now suppose that the electric potential $V$ satisfies ${\mathcal
D}_{\alpha}$ with $\alpha > 1$. Then relations \eqref{10} and
\eqref{sof2} -- \eqref{sof4} hold true with $X_1 = X_2 =
L^2({\mathcal S}_L)$, ${\mathcal H}_0 = H_0$, ${\mathcal H} = H$,
${\mathcal V} = V$, ${\mathcal K} = |V|^{1/2}$,  ${\mathcal J} = J
= {\rm sign}\,V$, and $\gamma = \gamma' = 1$.
    For $z \in {\mathbb C}_+$ set
$$
T(z): = |V|^{1/2}(H_{0} - z)^{-1}|V|^{1/2}.
$$
 By Lemma \ref{sofl1}
for almost every $E \in \re$ the operator-norm  limit
    \bel{sof5} T(E): =
    {\rm n}-\lim_{\delta \downarrow 0} T(E + i\delta)
    \ee
exists, and
    \bel{sof6}
    0 \leq {\rm Im}\,T(E) \in S_1.
    \ee
In Corollary \ref{af1} below we will show that the limit
\eqref{sof5} exists, and relation \eqref{sof6} holds true for {\em
every} $E \in \re\setminus {\mathcal Z}$. Then Theorem \ref{soft1}
implies that for almost every $E \in \re$ we have
\begin{equation} \label{sof7}
\xi(E ; H, H_0) = \int_{\re} {\rm ind}(J +  {\rm Re} \, T(E) + t\;
{\rm Im} \, T(E), J)\; d\mu (t),
\end{equation}
the right-hand-side being well defined for every $E \in
\re\setminus {\mathcal Z}$. In this article we identify the SSF
$\xi(E;H,H_0)$ for energies $E \not \in {\mathcal Z}$ with the
r.h.s. of \eqref{sof7}.\\

{\bf 4.4.} Fix $j \in {\mathbb N}$. Denote by $\varphi_j :
[0,\infty) \to [0,\infty)$ the function inverse to $E_j -
{\mathcal E}_j$. In the following lemma we describe some
properties of $\varphi_j$ which will be used in the sequel.
\\   Let $\beta$ and $\eta$ be two functions with values in
$[0,\infty)$, and ${\mathcal O} \subseteq D(\beta) \cap D(\eta)$.
We will write $\beta(s) \asymp \eta(s)$, $s \in {\mathcal O}$, if
there exist two constants $c_{\pm} > 0$ such that for each $s \in
{\mathcal O}$ we have $c_- \eta(s) \leq \beta(s) \leq c_+
\eta(s)$.
\begin{lemma} \label{fl}
Let $j \in {\mathbb N}$. We have
    \bel{f11}
    \varphi_j(s) \asymp s^{1/2}, \quad s \in [0,\infty),
    \ee
    \bel{san3a}
    \varphi'_j(s) \asymp s^{-1/2}, \quad s \in (0,\infty).
    \ee
Moreover,
    \bel{f16}
    \varphi_j(s) = \sqrt{s} \Phi(s), \quad s \in [0,\infty),
    \ee
where $\Phi \in C^{\infty}([0,\infty))$, and
    \bel{f17}
    \Phi(0) = \mu_j^{-1/2},
    \ee
the number $\mu_j$ being defined in \eqref{F1}. In particular, we
have
    \bel{f18}
    |\varphi_j''(s)| = O(s^{-3/2}), \quad s \in (0,s_0), \quad s_0 \in
    (0, \infty).
    \ee
\end{lemma}
\begin{proof}
By \eqref{a2} and \eqref{a5} we have
    \bel{f12}
    E_j(k) - {\mathcal E}_j \asymp k^2, \quad k \in \re,
    \ee
which implies immediately \eqref{f11}. On the other hand,
\eqref{f10} and \eqref{m21} easily yield
    \bel{f13}
    E_j'(k) \asymp k, \quad k
    \in [0,\infty).
    \ee
Bearing in mind the formula for the derivative of an  inverse
function, we find that \eqref{f11} and \eqref{f13} imply
\eqref{san3a}. \\
Further, for $t \geq 0$ introduce the function $E_j(\sqrt{t}) -
{\mathcal E}_j$, and denote by $\Psi = \Psi_j : [0,\infty) \to
[0,\infty)$ its inverse. By \eqref{san3a} we have $\Psi'(s) \asymp
1$, $s \in [0, \infty)$. Since $E_j$ is  analytic, we find that
$\Psi \in C^{\infty}([0,\infty))$. Moreover, $\Psi(0) = 0$ and
$\Psi'(0) = \mu_j^{-1}$. Since $\varphi(s) = \sqrt{\Psi(s)}$, we
get \eqref{f16} with $\Phi(s) = \sqrt{\Psi(s)/s}$, which on its
turn implies  \eqref{f17}.
\end{proof}
 For  $j \in
{\mathbb N}$ set
$$
 P_j : = {\mathcal F}^* \int_{\re}^{\oplus} p_j(k) dk {\mathcal F},
$$
the orthogonal projections $p_j(k)$, $k\in \re$, being defined in
\eqref{m32}. For $z \in {\mathbb C}_+$ and $j \in {\mathbb N}$ put
 $$
 T_j(z) : = |V|^{1/2} P_j (H_0-z)^{-1} |V|^{1/2}.
$$
\begin{lemma} \label{sanl1} Assume that $V$ satisfies ${\mathcal
D}_{\alpha}$ with $\alpha > 1$. Fix $j \in {\mathbb N}$. Then for
each $z \in {\mathbb C}_+$ we have $T_j(z) \in S_1$ and the
operator-valued function $T_j : {\mathbb C}_+ \to S_1$ is
analytic. Moreover, for $E \in \re\setminus\{{\mathcal E_j}\}$ the
limit
    \bel{san1}
T_j(E) = \lim_{\delta \downarrow 0} T_j(E+i\delta)
    \ee
exists in $S_1$, and  $T_j : \re\setminus\{{\mathcal E_j}\} \to
S_1$ is continuous. Next, if $E-{\mathcal E_j} < 0$, then the
operator $T_j(E)$ is self-adjoint, and if $E-{\mathcal E_j} > 0$,
we have
    \bel{a15}
    0 \leq {\rm Im}\,T_j(E), \quad {\rm rank}\,{\rm Im}\,T_j(E) \leq 2.
    \ee
Finally, for each $\lambda_0
> 0$ there exists $C_j = C_j(\lambda_0)$ such that for $0 <
|E-{\mathcal E}_j| < \lambda_0$ we have
    \bel{san0}
\|T_j(E)\|_1 \leq C_j |E-{\mathcal E}_j|^{-1/2};
    \ee
    if $E-{\mathcal E_j} < 0$, then $C_j$ could be chosen
    independent of $\lambda_0$.
\end{lemma}
    \begin{proof}
    Let $G = G_j :
\re \to S_2(L^2({\mathcal S}_L), {\mathbb C})$ be the
operator-valued function given for $k \in \re$ by
$$
G(k)u : = \frac{1}{\sqrt{2\pi}} \int_{\re} \int_{I_L} e^{-iky}
|V(x,y)|^{1/2} \psi_j(x; k) u(x,y) dx dy, \quad u \in
L^2({\mathcal S}_L).
$$
Evidently,
    \bel{a1}
\|G(k)^*G(k)\|_1 = \|G(k)\|_2^2 \leq c_1 : = \frac{1}{2\pi}
\sup_{x \in I_L}\int_{\re} |V(x,y)| dy
    \ee
for any $k \in \re$. Next,
$$
\|G(k_1) - G(k_2)\|_2^2 = \frac{1}{2\pi} \int_{\re} \int_{ I_L}
|V(x,y)| |e^{-ik_1 y} \psi_j(x; k_1) - e^{-ik_2 y} \psi_j(x;
k_2)|^2 dx dy \leq
$$
$$
\frac{2^{  2(1-\gamma)}}{ \pi} \sup_{x \in I_L} \int_{\re}
|V(x,y)| |y|^{2\gamma} dy |k_1 - k_2|^{2\gamma} + 2c_1 \int_{I_L}
|\psi(x; k_1) - \psi(x;k_2)|^2 dx
$$
for $k_1, k_2 \in \re$, and $\gamma \in (0,(\alpha-1)/2)$, $\gamma
\leq 1$. Since $\psi \in C^{\infty}(\re_k; L^2(I_L))$, we have
$$
\int_{I_L} |\psi(x; k_1) - \psi(x;k_2)|^2 dx = O(|k_1 - k_2|^2)
$$
for $k_1, k_2 \in (-k_0,k_0)$ with $k_0 \in (0,\infty)$.
Therefore,
    \bel{san6}
    \|G(k_1) - G(k_2)\|_2 = O(|k_1 - k_2|^{\gamma})
    \ee
for $k_1, k_2 \in (-k_0,k_0)$, $k_0 \in (0,\infty)$, and $\gamma
\in (0,(\alpha-1)/2)$, $\gamma \leq 1$. Taking into account
\eqref{a1} and \eqref{a2}, we find that if $z \in {\mathbb C}_+$,
then
    \bel{a4}
    \|G_j^* G_j (E_j - z)^{-1}\|_1 \in L^1(\re).
    \ee
 Then the spectral theorem implies
    \bel{san0a}
T_j(z) = \int_{\re} \frac{G_j(k)^* G_j(k)}{E_j(k) -z} dk, \quad z
\in {\mathbb C}_+,
    \ee
where, due to \eqref{a4} and the continuity of the functions $G_j
: \re \to S_2(L^2({\mathcal S}_L), {\mathbb C})$ and $E_j : \re
\to \re$, the integral admits an interpretation as a Bochner
integral
in the Banach space $S_1$ (see e.g. \cite{BI}), and it is easy to see that
$T_j : {\mathbb C}_+ \to S_1$ is analytic.  \\
 Let
$F = F_j : (0,\infty) \to S_2(L^2({\mathcal S}_L), {\mathbb C}^2)$
be the operator-valued function defined for $s \in (0,\infty)$ by
$$
F(s)u : = \sqrt{\varphi'(s)} (G(\varphi(s))u, G(-\varphi(s))u),
\quad u \in L^2({\mathcal S}_L),
$$
where, as above, $\varphi = \varphi_j  $ denotes the function
inverse to $E_j - {\mathcal E}_j  $. Then we have
$$
T_j(z) = \int_0^{\infty} \frac{F_j(s)^* F_j(s)}{s-\lambda-i\delta}
ds, \quad z = {\mathcal E}_j + \lambda + i\delta \in {\mathbb
C}_+.
$$
Further, if $\lambda: = E-{\mathcal E}_j < 0$, set
    \bel{a26aa}
T_j(E) = \int_0^{\infty} \frac{F_j(s)^* F_j(s)}{s-\lambda} ds.
    \ee
Evidently, the operator $T_j(E)$ is self-adjoint. Also, it is easy
to check that \eqref{san1} holds true, and the function $T_j :
(-\infty, {\mathcal E}_j) \to S_1$ is continuous. By \eqref{a1}
and \eqref{san3a},
$$
\|T_j(E)\|_1 \leq 2c_1 \int_0^{\infty}
\frac{\varphi'(s)}{s+|\lambda|} ds = O\left(\int_0^{\infty}
\frac{ds}{s^{1/2}(s+|\lambda|)}\right) = O(|\lambda|^{-1/2}),
\quad \lambda < 0,
$$
so that \eqref{san0} holds in this case as well. \\
Let now $\lambda = E-{\mathcal E}_j > 0$. For $E = {\mathcal E}_j
+ \lambda$ put
    \bel{a17}
{\rm Re}\,T_j(E) : = {\rm v.p.} \, \int_0^{\infty} \frac{F_j(s)^*
F_j(s)}{s-\lambda} ds,
    \ee
    \bel{a18}
 {\rm Im}\,T_j(E) : = \pi F_j(\lambda)^*
F_j(\lambda), \ee
$$
T_j(E) : = {\rm Re}\,T_j(E) + i {\rm Im}\,T_j(E).
$$
Note that \eqref{a18} immediately implies \eqref{a15}. Moreover,
$$
{\rm v.p.} \, \int_0^{\infty} \frac{F(s)^* F(s)}{s-\lambda} ds =
\int_0^{\lambda/2} \frac{F(s)^* F(s)}{s-\lambda} ds +
\int_{3\lambda/2}^{\infty} \frac{F(s)^* F(s)}{s-\lambda} ds +
$$
    \bel{a29}
\int_{0}^{\lambda/2} \left(F(\lambda + \nu)^* F(\lambda + \nu) -
F(\lambda - \nu)^* F(\lambda - \nu)\right) \frac{d\nu}{\nu}.
    \ee
By \eqref{san3a} and \eqref{f18},
    \bel{san4}
|\varphi(\lambda + \nu) - \varphi(\lambda - \nu)| = O\left(
 (\lambda + \nu)^{1/2} - (\lambda - \nu)^{1/2}\right),
    \ee
    \bel{san5}
\left|\varphi'(\lambda + \nu) - \varphi'(\lambda - \nu)\right| =
O\left(
 (\lambda - \nu)^{-1/2} - (\lambda + \nu)^{-1/2}\right),
    \ee
for $\nu \in (0, \lambda/2)$, $\lambda \in (0,\lambda_0)$.\\
Taking into account \eqref{a1} - \eqref{san6}, \eqref{san3a}, and
   \eqref{a17} - \eqref{san5} we find that the operator
$T_j(E)$ is well defined, that \eqref{san1} holds true again, and
$$
\|F(\lambda)^* F(\lambda)\|_1 = O(\lambda^{-1/2}), \quad \lambda >
0,
$$
$$
\left\|\int_0^{\lambda/2} \frac{F(s)^* F(s)}{s-\lambda} ds
\right\|_1 = O(\lambda^{-1/2}), \quad
\left\|\int_{3\lambda/2}^{\infty} \frac{F(s)^* F(s)}{s-\lambda} ds
\right\|_1 = O(\lambda^{-1/2}), \quad \lambda > 0,
$$
$$
\left\|\int_{0}^{\lambda/2}  (F(\lambda + \nu)^* F(\lambda + \nu)
- F(\lambda - \nu)^* F(\lambda - \nu)) \frac{d\nu}{\nu}\right\|_1
= O(\lambda^{-1/2}), \quad \lambda \in (0,\lambda_0),
$$
 which yields again \eqref{san0}.
    \end{proof}
{\bf 4.5.} Let $j \in {\mathbb N}$. Set $P^+_j : =
\sum_{m=j}^{\infty}P_m$ where the convergence of the infinite sum
is understood in the strong sense. For $z \in {\mathbb C}_+$,
${\rm Re}\, z < {\mathcal E}_j$, put
$$
T_j^+(z): = |V|^{1/2} P_j^+ (H_0-z)^{-1} |V|^{1/2}.
$$
\begin{lemma} \label{al1}
Fix $j \in {\mathbb N}$. Let $E \in (-\infty, {\mathcal E}_j)$.
Then the limit
    \bel{a9}
    T_j^+(E) = T_j^+(E)^* = {\rm n}-\lim_{\delta \downarrow 0} T_j^+(E+i\delta)
    \ee
exists. Moreover, for any $z \in \overline{{\mathbb
C}_+}\setminus[{\mathcal E}_j,\infty)$ we have $T_j^+(z) \in S_2$,
and the operator-valued function $T_j^+ : \overline{{\mathbb
C}_+}\setminus [{\mathcal E}_j,\infty) \to S_2$ is continuous.
Finally, there exists a constant $C_+$ which depends on $V$, but
is independent of $E$ and $j$, such that
    \bel{a15a}
    \|T_j^+(E)\|_2 \leq C_+ {\mathcal E}_j ({\mathcal E}_j -
    E)^{-1}, \quad E \in (-\infty, {\mathcal E}_j).
    \ee
\end{lemma}
\begin{proof}
We have
    \bel{a10}
    P_j^+ (H_0 - z)^{-1} = P_j^+ {\mathbb P}_{[{\mathcal E}_j,
    \infty)}(H_0) (H_0 - z)^{-1}
    \ee
and the operator valued function ${\mathbb P}_{[{\mathcal E}_j,
    \infty)}(H_0) (H_0 - z)^{-1}$ is analytic even on ${\mathbb
    C}\setminus  [{\mathcal E}_j,\infty)$. Since $P_j^+$ and
    $|V|^{1/2}$ are bounded operators, this analyticity implies,
    in particular,
    the existence of the limit in \eqref{a9} and the continuity of
$T_j^+ : \overline{{\mathbb C}_+}\setminus [{\mathcal E}_j,\infty)
\to {\mathcal B}$. Further,
    \bel{a11}
    \||V|^{1/2}P_j^+ (H_0 - E)^{-1}|V|^{1/2}\|_2 \leq
    \sup_{(x,y) \in {\mathcal S}_L} |V(x,y)|^{1/2} \|P_j^+ (H_0 -
    E)^{-1}H_0\| \, \|H_0^{-1} |V|^{1/2}\|_2.
    \ee
By \eqref{a10},
    \bel{a12}
\|P_j^+ (H_0 - E)^{-1}H_0\| \leq \|{\mathbb P}_{[{\mathcal E}_j,
    \infty)}(H_0) (H_0 - E)^{-1}H_0\| \leq \sup_{\lambda \in [{\mathcal E}_j,
    \infty)} \lambda (\lambda - E)^{-1} = {\mathcal E}_j ({\mathcal E}_j -
    E)^{-1}.
    \ee
On the other hand, the diamagnetic inequality for Hilbert-Schmidt
operators (see e.g. \cite[Theorem 2.13]{S}) implies
    \bel{a13}
\|H_0^{-1} |V|^{1/2}\|_2 \leq \|\Delta_D^{-1} |V|^{1/2}\|_2
    \ee
where, as above, $\Delta_D$ is the Dirichlet Laplacian defined on
${\mathcal S}_L$. The integral kernel of $\Delta_D$ is explicitly
known, and we easily find
    \bel{a14}
\|\Delta_D^{-1} |V|^{1/2}\|_2^2 \leq 16 c_1 \frac{L^3}{\pi^3}
\sum_{n=1}^{\infty} n^{-3} \int_0^{\infty}
\frac{d\xi}{(\xi^2+1)^2}.
    \ee
Putting together \eqref{a11} - \eqref{a14}, we obtain \eqref{a15a}. \\
Finally, an estimate similar to \eqref{a11} of the Hilbert-Schmidt
norm of the difference $|V|^{1/2}P_j^+ (H_0 - z_1)^{-1}|V|^{1/2} -
|V|^{1/2}P_j^+ (H_0 - z_2)^{-1}|V|^{1/2}$, $z_1, z_2 \in
\overline{{\mathbb C}_+}\setminus [{\mathcal E}_j,\infty)$ easily
implies the continuity of $T_j^+ : \overline{{\mathbb
C}_+}\setminus [{\mathcal E}_j,\infty) \to S_2$.
\end{proof}
{\bf 4.6.} In this subsection we prove \eqref{sof5} - \eqref{sof6}
as well as Proposition \ref{ap21}. \\
Let $E \in \re\setminus{\mathcal Z}$. If $E$ has one nearest
element from ${\mathcal Z}$, let $q = q(E)$ be the number of this
neighbour; if $E$ has two nearest elements from ${\mathcal Z}$,
for definiteness let $q(E)$ be the number of the greater of these
elements. Set
    \bel{a16}
    T(E): = \sum_{j=1}^{q(E)}T_j(E) + T_{q(E)+1}^+(E).
    \ee
\begin{follow} \label{af1}
Let $V$ satisfy ${\cal D}_{\alpha}$ with $\alpha > 1$, and let $E
\in \re\setminus{\mathcal Z}$. Then \eqref{sof5} - \eqref{sof6}
hold true, the limiting operator $T(E)$ being defined in
\eqref{a16}. Moreover,
    \bel{a18a}
    {\rm rank}\;{\rm Im}\,T(E) \leq 2 q(E).
    \ee
\end{follow}
\begin{proof}
In order to prove the existence   of the limit \eqref{sof5}, we
just have to write
$$
T(E+i\delta): = \sum_{j=1}^{q(E)}T_j(E+i\delta) +
T_{q(E)+1}^+(E+i\delta), \quad \delta > 0,
$$
and to apply \eqref{san1} and \eqref{a9}. In order to prove
\eqref{sof6} and \eqref{a18a}, it suffices to apply \eqref{a15},
bearing in mind that ${\rm Im}\, T(E) = \sum_{j=1}^{q(E)} {\rm
Im}\, T_j(E)$.
    \end{proof}
    Next we prove Proposition \ref{ap21}. The
proof of the continuity of the SSF repeats word by word the proof
of the continuity part of \cite[Proposition 2.5]{BPR}. Let us show
that the SSF is locally bounded, i.e.  that it is bounded on every
compact
subset of $\re\setminus{\mathcal Z}$. \\
Let $E \in \re\setminus{\mathcal Z}$. Applying \eqref{sof7},
\eqref{a20}, and \eqref{a20a}, we get
    \bel{a22}
    |\xi(E;H,H_0)| \leq n_*(s; {\rm Re}\,T(E)) + {\rm rank}\,{\rm
    Im}\,T(E), \quad s \in (0,1).
    \ee
By \eqref{a19},
    \bel{a23}
n_*(s; {\rm Re}\,T(E)) \leq n_*(s/2; \sum_{j=1}^{q(E)}{\rm
Re}\,T_j(E)) + n_*(s/2; T^+_{q(E)+1}(E)).
    \ee
    Using \eqref{32a} with $p=1$ and $p=2$, as well as
    \eqref{san0} and \eqref{a15}, we get
    \bel{a24}
    n_*(s/2; \sum_{j=1}^{q(E)}{\rm
Re}\,T_j(E)) \leq \frac{2}{s} \sum_{j=1}^{q(E)}\|T_j(E)\|_1  \leq
\frac{2}{s} \sum_{j=1}^{q(E)} C_j |E - {\mathcal E}_j|^{-1/2},
    \ee
    \bel{a25a}
    n_*(s/2; T^+_{q(E)+1}(E)) \leq \frac{4}{s^2}
    \|T^+_{q(E)+1}(E)\|_2^2 \leq \frac{4}{s^2}
    C_+^2 {\mathcal E}_{q(E) + 1}^2 ({\mathcal E}_{q(E) + 1} -
    E)^{-2}.
    \ee
    Now the combination of \eqref{a22}, \eqref{a15}, and
    \eqref{a23} - \eqref{a25a} implies the local boundedness of the
    SSF. \\

{\bf 4.7}. In this subsection we prove Theorem \ref{t2}.
\begin{pr} \label{ap42} Assume that $V$ satisfies ${\cal D}_{\alpha}$
with $\alpha > 1$. Pick $q \in {\mathbb N}$ and $\lambda \neq 0$
such that $E: = {\mathcal E}_q + \lambda \not \in {\mathcal Z}$.
Then we have
    \bel{a28}
    {\rm ind}\,(J + \varepsilon + {\rm Re}\,T_q(E), J +
    \varepsilon) + O(1) \leq \xi(E;H,H_0) \leq
    {\rm ind}\,(J - \varepsilon + {\rm Re}\,T_q(E), J -
    \varepsilon) + O(1)
    \ee
as $\lambda \to 0$ for each $\varepsilon \in (0,1)$.
    \end{pr}
    \begin{proof}
Applying   \eqref{sof7}, \eqref{a20a}, and \eqref{a18a}, we get
    \bel{a25}
\left| \xi(E;H,H_0) -  {\rm ind}(J+{\rm Re}\,T(E), J)\right| \leq
2 q(E).
    \ee
Write ${\rm Re}\,T(E) = {\rm Re}\,T_q(E) + \tilde{T}(E)$ where
$\tilde{T}(E) : = \sum_{j<q} {\rm Re}\,T_j(E) + T_{q+1}^+(E).$ By
\eqref{paris2},
$$
{\rm ind}\,(J + \varepsilon + {\rm Re}\,T_q(E), J +
    \varepsilon) - n_*(\varepsilon; \tilde{T}(E)) \leq {\rm ind}(J+{\rm Re}\,T(E),
    J)\leq
    $$
    \bel{a26}
{\rm ind}\,(J - \varepsilon + {\rm Re}\,T_q(E), J - \varepsilon) +
n_*(\varepsilon; \tilde{T}(E)).
    \ee
Using \eqref{a19} and arguing as in the derivation of \eqref{a24},
\eqref{a25a}, we get
    \bel{a27}
n_*(\varepsilon; \tilde{T}(E)) \leq \frac{2}{\varepsilon} \sum_{j:
j<q} C_j |{\mathcal E}_q - {\mathcal E}_j + \lambda|^{-1/2} +
\frac{4}{\varepsilon^2}C_+^2 {\mathcal E}_{q+1}^2 ({\mathcal
E}_{q+1} - {\mathcal E}_q - \lambda)^{-2} = O(1), \quad \lambda
\to 0.
    \ee
Now the combination of \eqref{a25} -- \eqref{a27} yields
\eqref{a28}.
    \end{proof}
    Fix $j \in {\mathbb N}$. Let $g = g_j : \re \to S_2(L^2({\mathcal S}_L),
    {\mathbb C})$ be the operator-valued function given for $k \in
    \re$ by
$$
g_j(k)u = \frac{1}{\sqrt{2\pi}} \int_{\re} \int_{I_L} e^{-iky}
|V(x,y)|^{1/2} \psi_j(x;0) u(x,y) dx dy, \quad u \in L^2({\mathcal
S}_L).
$$
Similarly to \eqref{a1} and \eqref{san6} we have
    \bel{a31}
    \|g(k)\|_2^2 \leq c_1, \quad k \in \re,
    \ee
    \bel{a32}
    \|g(k_1) - g(k_2)\|_2 = O(|k_1 - k_2|^{\gamma}), \quad k_1,
    k_2 \in \re,
    \ee
for any $\gamma \in (0,(\alpha-1)/2)$ such that $\gamma \leq 1$.
By analogy with \eqref{san0a} set
    \bel{a38}
    \tilde{\tau}_j(z) : = \int_{\re} \frac{g_j(k)^* g_j(k)}{E_j(k) -
    z } dk, \quad z \in {\mathbb C}_+.
    \ee
As in the case of the operator $T_j(z)$ (see Lemma \ref{sanl1}) we
can show that in $S_1$ there exists a limit
$$
\tilde{\tau}_j(E) = \lim_{\delta \downarrow 0}
\tilde{\tau}_j(E+i\delta), \quad E \in \re\setminus\{{\mathcal
E}_j\}.
$$
\begin{pr} \label{ap43} Let $V$ satisfy ${\mathcal D}_{\alpha}$
with $\alpha > 1$. Fix $q \in {\mathbb N}$, and let $E = {\mathcal
E}_q + \lambda \not \in {\mathcal Z}$. Then for each $\varepsilon
\in (0,1/2)$ we have
    \bel{a27a}
{\rm ind}\,(J + 2\varepsilon + {\rm Re}\,\tilde{\tau}_q(E), J +
2\varepsilon) + O(1) \leq  {\rm ind}\,(J + \varepsilon + {\rm
Re}\,T_q(E), J + \varepsilon),
    \ee
    \bel{a28a}
{\rm ind}\,(J - 2\varepsilon + {\rm Re}\,\tilde{\tau}_q(E), J -
2\varepsilon) + O(1) \geq  {\rm ind}\,(J - \varepsilon + {\rm
Re}\,T_q(E), J - \varepsilon),
    \ee
as $\lambda \downarrow 0$.
\end{pr}
\begin{proof}
Using \eqref{paris2} and \eqref{a20}, we obtain
$$
{\rm ind}\,(J + 2\varepsilon + {\rm Re}\,\tilde{\tau}_q(E), J +
2\varepsilon) - n_*(\varepsilon; {\rm Re}\,T_q(E) - {\rm
Re}\,\tilde{\tau}_q(E)) \leq {\rm ind}\,(J + \varepsilon + {\rm
Re}\,T_q(E), J + \varepsilon),
$$
$$
{\rm ind}\,(J - 2\varepsilon + {\rm Re}\,\tilde{\tau}_q(E), J -
2\varepsilon) + n_*(\varepsilon; {\rm Re}\,T_q(E) - {\rm
Re}\,\tilde{\tau}_q(E)) \geq {\rm ind}\,(J - \varepsilon + {\rm
Re}\,T_q(E), J - \varepsilon).
    $$
Hence, in order to prove \eqref{a27a} -- \eqref{a28a}, it suffices
to show that for each $\varepsilon>0$ we have
    \bel{a37}
n_*(\varepsilon; {\rm Re}\,T_q(E) - {\rm Re}\,\tilde{\tau}_q(E)) =
O(1), \quad \lambda \to 0.
    \ee
Let again $\varphi = \varphi_q$ be the function inverse to $E_q -
{\mathcal E}_q$. Denote by $f = f_q: (0,\infty) \to
S_2(L^2({\mathcal S}_L), {\mathbb C}^2)$ the operator-valued
function defined for $s \in (0,\infty)$ by
$$
f(s)u : = \sqrt{\varphi'(s)} (g(\varphi(s))u, g(-\varphi(s))u),
\quad u \in L^2({\mathcal S}_L).
$$
Then similarly to \eqref{a26aa}, \eqref{a17}, and \eqref{a29}, we
have
$$
 \tilde{\tau}_q({\mathcal E}_q + \lambda) = \tilde{\tau}_q({\mathcal E}_q + \lambda)^* =
 \int_0^{\infty} \frac{f_q(s)^* f_q(s)}{s-\lambda} ds
$$
if $\lambda < 0$, and
    $$
{\rm Re}\,\tilde{\tau}_q({\mathcal E}_q + \lambda)   =
\int_0^{\lambda/2} \frac{f_q(s)^* f_q(s)}{s-\lambda} ds +
\int_{3\lambda/2}^{\infty} \frac{f_q(s)^* f_q(s)}{s-\lambda} ds +
$$
    $$
\int_{0}^{\lambda/2} \left(f_q(\lambda + \nu)^* f_q(\lambda + \nu)
- f_q(\lambda - \nu)^* f_q(\lambda - \nu)\right) \frac{d\nu}{\nu}
    $$
if $\lambda > 0$. Further, we have
$$
G(k) = g(k) + k \varrho(k)
$$
where $\varrho : \re \to L^2(L^2({\mathcal S}_L),{\mathbb C})$ is
the operator-valued function given for $k \in \re$ by
$$
\varrho(k)u = \frac{1}{\sqrt{2\pi}} \int_{\re} \int_{I_L} e^{-iky}
|V(x,y)|^{1/2} \tilde{\psi}(x;k) u(x,y) dx dy, \quad u \in
L^2({\mathcal S}_L),
$$
where $\tilde{\psi}(x;k) : = \frac{\psi_q(x;k) - \psi_q(x;0)}{k}$.
Evidently,
    \bel{a33}
    \|\varrho(k)\|_2^2 \leq c_1 \int_{I_L} \tilde{\psi}(x;k)^2 dx,
    \quad k \in \re,
    \ee
    \bel{a34}
\|\varrho(k_1) - \varrho(k_2)\|_2 = O(|k_1 - k_2|^{\gamma})
    \ee
    for $k_1,
    k_2 \in (-k_0,k_0)$ with $k_0 \in (0,\infty)$, and
    $\gamma \in (0, (\alpha-1)/2)$, $\gamma \leq 1$.
Next, we have
$$
F(s) = f(s) + \varphi(s) r(s)
$$
where $r : (0,\infty) \to S_2(L^2({\mathcal S}_L), {\mathbb C}^2)$
is the operator-valued function defined for $s \in (0,\infty)$ by
$$
r(s)u : = \sqrt{\varphi'(s)} (\varrho(\varphi(s))u,
-\varrho(-\varphi(s))u), \quad u \in L^2({\mathcal S}_L).
$$
Therefore,
    \bel{a40}
F(s)^*F(s) = f(s)^*f(s) + 2\varphi(s) {\rm Re}\, f(s)^*r(s) +
\varphi(s)^2 r(s)^*r(s).
    \ee
    Note that
$$
\varphi(s) f(s)^* r(s) = \varphi(s)
\varphi'(s)\left(g(\varphi(s))^* \varrho(\varphi(s)) -
g(-\varphi(s))^* \varrho(-\varphi(s))\right) =
$$
$$
\varphi(s) \varphi'(s)\left(g(\varphi(s))^* (\varrho(\varphi(s)) -
\varrho(-\varphi(s))) + (g(\varphi(s))^*  - g(-\varphi(s))^*)
\varrho(-\varphi(s))\right).
$$
Hence, by \eqref{f11} -- \eqref{san3a},  \eqref{a31} --
\eqref{a32}, and \eqref{a33} -- \eqref{a34}, we have
    \bel{a35}
    \varphi(s)\|f(s)^*r(s) \|_1 = O(s^{\gamma/2}), \quad \gamma \in
    (0,(\alpha-1)/2), \quad \gamma \leq 1,
    \ee
    \bel{a36}
    \varphi(s)^2 \|r(s)^*r(s)\|_1 = O(s^{1/2})
    \ee
for $s \in (0,s_0)$ and $s_0 \in (0,\infty)$. By \eqref{a40}, for
a fixed $s_0 > 0$ we have
$$
{\rm Re}\,T_q({\mathcal E}_q + \lambda) - {\rm
Re}\,\tilde{\tau}_q({\mathcal E}_q + \lambda) = T_q({\mathcal E}_q
+ \lambda) - \tilde{\tau}_q({\mathcal E}_q + \lambda) =
$$
$$
\int_{s_0}^{\infty} \frac{F(s)^* F(s)}{s-\lambda} ds -
\int_{s_0}^{\infty} \frac{f(s)^* f(s)}{s-\lambda} ds +
\int_0^{s_0} \frac{2\varphi(s) {\rm Re}\, f(s)^*r(s) +
\varphi(s)^2 r(s)^*r(s)}{s-\lambda} ds
$$
if $\lambda < 0$, and
$$
{\rm Re}\,T_q({\mathcal E}_q + \lambda) - {\rm
Re}\,\tilde{\tau}_q({\mathcal E}_q + \lambda) =
\int_{s_0}^{\infty} \frac{F(s)^* F(s)}{s-\lambda} ds -
\int_{s_0}^{\infty} \frac{f(s)^* f(s)}{s-\lambda} ds +
$$
$$
\int_{3\lambda/2}^{s_0} \frac{2\varphi(s) {\rm Re}\, f(s)^*r(s) +
\varphi(s)^2 r(s)^*r(s)}{s-\lambda} ds + \int_0^{\lambda/2}
\frac{2\varphi(s) {\rm Re}\, f(s)^*r(s) + \varphi(s)^2
r(s)^*r(s)}{s-\lambda} ds +
$$
$$
2\int_{0}^{\lambda/2} (\varphi(\lambda + \nu) {\rm Re}\, f(\lambda
+ \nu)^*r(\lambda + \nu) - \varphi(\lambda - \nu) {\rm Re}\,
f(\lambda - \nu)^*r(\lambda - \nu)) \frac{d\nu}{\nu} +
$$
$$
\int_{0}^{\lambda/2} (\varphi(\lambda + \nu)^2 r(\lambda +
\nu)^*r(\lambda + \nu) - \varphi(\lambda - \nu)^2 r(\lambda -
\nu)^*r(\lambda - \nu)) \frac{d\nu}{\nu}
$$
if $\lambda$ is positive and small enough (say, $\lambda \in
(0,s_0/2)$). Using estimates \eqref{san4} - \eqref{san5} as well
as \eqref{a31} - \eqref{a32}, \eqref{a33} - \eqref{a34}, and
\eqref{a35} - \eqref{a36}, we obtain
$$
\|{\rm Re}\,T_q({\mathcal E}_q + \lambda) - {\rm
Re}\,\tilde{\tau}_q({\mathcal E}_q + \lambda)\|_1 = O(1), \quad
\lambda \to 0,
$$ which combined with \eqref{32a} for $p=1$ yields \eqref{a37},
and hence \eqref{a27a} - \eqref{a28a}.
\end{proof}
Fix $j \in {\mathbb N}$. By analogy with \eqref{san0a} and
\eqref{a38} set
$$
\tau_j(z) : = \int_{\re} \frac{g(k)^* g(k)}{\mu_j k^2 - z} dk,
\quad z \in {\mathbb C}_+.
$$
As in the case of the operators $T_j(z)$ and $\tilde{\tau}(z)$, in
$S_1$ there exists a limit
$$
\tau_j(E) = \lim_{\delta \downarrow 0} \tau_j(E + i\delta), \quad
E \in \re \setminus \{0\}.
$$
\begin{pr} \label{ap44}
Let $V$ satisfy ${\cal D}_{\alpha}$ with $\alpha > 1$. Fix $q \in
{\mathbb N}$. Then for each $\varepsilon \in (0,1/2)$ and $\gamma
\in (0, (\alpha - 1)/2)$, $\gamma \leq 1$, we have
    \bel{m10}
{\rm ind}\,(J + 2\varepsilon + {\rm Re}\tau_q(\lambda), J +
2\varepsilon) + O(\theta_{2\gamma}(\lambda)) \leq {\rm ind}\,(J +
\varepsilon + {\rm Re}\tilde{\tau}_q({\mathcal E}_q + \lambda), J
+ \varepsilon),
    \ee
    \bel{m11}
{\rm ind}\,(J - 2\varepsilon + {\rm Re}\tau_q(\lambda), J -
2\varepsilon) + O(\theta_{2\gamma}(\lambda)) \geq {\rm ind}\,(J -
\varepsilon + {\rm Re}\tilde{\tau}_q({\mathcal E}_q + \lambda), J
- \varepsilon),
    \ee
as $\lambda \to 0$, the functions $\theta_{\beta}$ being defined
in \eqref{m7} -- \eqref{m8}.
\end{pr}
\begin{proof}
Similarly to the proof of Proposition \ref{ap43} (see
\eqref{a37}), it suffices to show that for each $\varepsilon > 0$
we have
    \bel{m9}
n_*(\varepsilon; {\rm Re}\,\tilde{\tau}_q({\mathcal E}_q +
\lambda) - {\rm Re}\,\tau_q(\lambda)) =
O(\theta_{2\gamma}(\lambda)), \quad \lambda \to 0.
    \ee
Let at first $\lambda < 0$. In this case we have
$$
{\rm Re}\,\tilde{\tau}_q({\mathcal E}_q + \lambda) - {\rm
Re}\,\tau_q(\lambda) = \tilde{\tau}_q({\mathcal E}_q + \lambda) -
\tau_q(\lambda) =
$$
$$
\int_{\re} g_q(k)^* g_q(k)\frac{\mu_q k^2 - E_q(k) + {\mathcal
E}_q}{(E_q(k) - {\mathcal E}_q - \lambda)(\mu_q k^2 - \lambda)} dk
$$
and
$$
\| \tilde{\tau}({\mathcal E}_q + \lambda) - \tau_q( \lambda)\|_1
\leq \frac{c_1}{\mu_q} \int_{\re} \frac{|E_q(k) - {\mathcal E}_q -
\mu_q k^2|}{k^2 (E_q(k) - {\mathcal E}_q)} dk = O(1), \quad
\lambda \uparrow 0,
$$
which combined with \eqref{32a} for $p=1$ yields \eqref{m9} in the
case $\lambda < 0$. \\
Let now $\lambda > 0$. As above, let $\varphi = \varphi_q$ be the
function inverse to $E_q - {\mathcal E}_q$. Set
$$
\phi(s) = \phi_q(s): = \mu_q^{-1/2} s^{1/2}, \quad s>0.
$$
By \eqref{f16} - \eqref{f17},
    \bel{m12}
    \varphi(s) - \phi(s) = O(s^{3/2}),
    \ee
    \bel{m13}
    \varphi'(s) - \phi'(s) = O(s^{1/2}),
    \ee
for $s \in (0,s_0)$ and $s_0 \in (0,\infty)$. Fix $s_0 \in
(0,\infty)$ and assume $\lambda < s_0/2$. For $\eta = \varphi$ or
$\eta = \phi$  define the operator-valued function $\Gamma_{\eta}
: (0,\infty) \to S_2(L^2({\mathcal S}_L), {\mathbb C}^2)$ by
$$
\Gamma_{\eta}(s)u : = (g(\eta(s))u, g(-\eta(s))u), \quad s > 0,
\quad u \in L^2({\mathcal S}_L),
$$
and
$$
M_{\eta, 1}(\lambda) : = \int_{s_0}^{\infty} \eta'(s)
\frac{\Gamma_{\eta}(s)^*\Gamma_{\eta}(s)}{s-\lambda} ds,
$$
$$
M_{\eta, 2}(\lambda) : =  {\rm v.p.}\,\int_0^{s_0} \eta'(s)
\frac{2 {\rm Re}\,\Gamma_{\eta}(s)^*\Gamma_{\eta}(\lambda) -
\Gamma_{\eta}(\lambda)^*\Gamma_{\eta}(\lambda)}{s-\lambda} ds,
$$
$$
M_{\eta, 3}(\lambda) : = \int_0^{s_0} \eta'(s)
\frac{(\Gamma_{\eta}(s)-\Gamma_{\eta}(\lambda))^*(\Gamma_{\eta}(s)
- \Gamma_{\eta}(\lambda))}{s-\lambda} ds.
$$
Then we have
$$
\tilde{\tau}_q({\mathcal E}_q + \lambda) = \sum_{l=1,2,3}
M_{\varphi, l}(\lambda), \quad \tau_q(\lambda) = \sum_{l=1,2,3}
M_{\phi, l}(\lambda).
$$
It is easy to see that
    \bel{m4}
    \|M_{\eta, 1}(\lambda)\|_1 = O(1), \quad \lambda \downarrow 0,
\quad \eta = \varphi, \phi,
    \ee
    \bel{m5}
    {\rm rank}\,M_{\eta, 2}(\lambda) \leq 6, \quad \lambda > 0,
\quad \eta = \varphi, \phi,
    \ee
    \bel{m33}
    \|M_{\eta, 3}(\lambda)\|_1 = O(\theta_{\gamma}(\lambda)), \quad \lambda \downarrow 0,
\quad \eta = \varphi, \phi.
    \ee
Let us show that
    \bel{m3}
    \|M_{\varphi, 3}(\lambda) - M_{\phi, 3}(\lambda)\|_1 =
    O(\theta_{2\gamma}(\lambda)), \quad \lambda \downarrow 0.
    \ee
We have
$$
M_{\varphi, 3}(\lambda) - M_{\phi, 3}(\lambda) =
$$
$$
\int_0^{s_0} (\varphi'(s) - \phi'(s))
\frac{(\Gamma_{\varphi}(s)-\Gamma_{\varphi}(\lambda))^*(\Gamma_{\varphi}(s)
- \Gamma_{\varphi}(\lambda))}{s-\lambda} ds +
$$
$$
\int_0^{s_0} \phi'(s) \frac{(\Gamma_{\varphi}(s) -
\Gamma_{\phi}(s) -\Gamma_{\varphi}(\lambda) +
\Gamma_{\phi}(\lambda))^*(\Gamma_{\varphi}(s) -
\Gamma_{\varphi}(\lambda))}{s-\lambda} ds +
$$
$$
\int_0^{s_0} \phi'(s) \frac{(\Gamma_{\phi}(s)
-\Gamma_{\phi}(\lambda))^*(\Gamma_{\varphi}(s) - \Gamma_{\phi}(s)
- \Gamma_{\varphi}(\lambda) + \Gamma_{\phi}(\lambda))}{s-\lambda}
ds : =
$$
$$
I_1 + I_2 + I_3.
$$
Using \eqref{m13}, \eqref{a32}, and \eqref{san3a} which implies
$|\varphi(s) - \varphi(\lambda)| = O(|\sqrt{s} -
\sqrt{\lambda}|)$, $s \in (0,s_0)$, we get
    \bel{m1}
\|I_1\|_1 = O\left( \int_0^{s_0} s^{1/2} \frac{|\sqrt{s} -
\sqrt{\lambda}|^{2\gamma}}{|s-\lambda|} ds \right) = O(1), \quad
\lambda \downarrow 0.
    \ee
Further, for $s, \lambda > 0$, and $\gamma \in (0, (\alpha-1)/2)$,
$\gamma \leq 1$, we have
$$
\|\Gamma_{\varphi}(s) - \Gamma_{\phi}(s)
-\Gamma_{\varphi}(\lambda) + \Gamma_{\phi}(\lambda)\|^2_2 \leq
$$
$$
\frac{1}{\pi} \sup_{(x,y) \in {\mathcal S}_L} \langle y
\rangle^{\alpha} |V(x,y)| \int_{\re} |e^{i\varphi(s)y} -
e^{i\phi(s)y} - e^{i\varphi(\lambda)y} + e^{i\phi(\lambda)y}|^2
\langle y \rangle^{-\alpha} dy \leq
$$
$$
\frac{2^{  3-2\gamma}}{\pi} \sup_{(x,y) \in {\mathcal S}_L}
\langle y \rangle^{\alpha} |V(x,y)| \int_{\re} |y|^{2\gamma}
\langle y \rangle^{-\alpha} dy \left(|\varphi(s) -
\phi(s)|^{2\gamma} + |\varphi(\lambda) -
\phi(\lambda)|^{2\gamma}\right).
$$
Using \eqref{m12}, we get
    \bel{m2} \|I_j\|_1 = O\left(
\int_0^{s_0} s^{-1/2} \frac{(s^{3\gamma} +
\lambda^{3\gamma})^{1/2}|\sqrt{s} -
\sqrt{\lambda}|^{\gamma}}{|s-\lambda|} ds \right) =
O(\theta_{2\gamma}(\lambda)), \quad \lambda \downarrow 0, \quad
j=2,3.
    \ee
Putting together \eqref{m1} and \eqref{m2}, we obtain \eqref{m3}.
Now the combination of \eqref{m4} -- \eqref{m3} with \eqref{m15}
and \eqref{32a} for $p=1$ yields \eqref{m9} in the case $\lambda >
0$.
\end{proof}
Next, we note that for each $\lambda > 0$ and $q \in {\mathbb N}$
we have ${\rm rank}\,{\rm Im}\, \tau_q(\lambda) \leq 2$, while
${\rm Im}\, \tau_q(\lambda) = 0$ if $\lambda < 0$. Therefore,
    \bel{m16}
{\rm ind}\,(J-\varepsilon + {\rm Re}\,\tau_q(\lambda),
J-\varepsilon) = \int_{\re} {\rm ind}\,(J-\varepsilon + {\rm
Re}\,\tau_q(\lambda) + t {\rm Im}\,\tau_q(\lambda), J-\varepsilon)
d\mu(t) + O(1), \quad \lambda \to 0,
    \ee
for each $\varepsilon \in (-1,1)$. On the other hand, we have
$$
w_{q,\varepsilon} = \varkappa^* (J - \varepsilon)^{-1} \varkappa,
\quad \varepsilon \in (-1,1),
$$
$$
\tau_q(z) = \varkappa (h_{0,q} - z)^{-1} \varkappa^*, \quad z \in
\overline{{\mathbb C}_+}\setminus \{0\},
$$
where $\varkappa : L^2(\re) \to L^2({\mathcal S}_L)$ is the
operator defined by
$$
(\varkappa u)(x,y) : = \psi(x,0) |V(x,y)|^{1/2} u(y), \quad u \in
L^2(\re).
$$
By Theorem \ref{soft1} we have
    \bel{m17}
\int_{\re} {\rm ind}\,(J-\varepsilon + {\rm Re}\,\tau_q(\lambda) +
t {\rm Im}\,\tau_q(\lambda), J-\varepsilon) d\mu(t) = \xi(\lambda;
h_q(\varepsilon), h_{0,q}), \quad \lambda \neq 0.
    \ee
Combining \eqref{a28}, \eqref{a27a} -- \eqref{a28a}, \eqref{m10}
-- \eqref{m11}, \eqref{m16}, and  \eqref{m17}, we obtain
\eqref{m18}.\\

{\bf 4.8.} In this subsection we give a sketch of the proof of
Corollary \ref{f1}. Let $w = \bar{w} \in L^{\infty}(\re)$. Set
$$
h_0 : = -\frac{d^2}{dy^2}, \quad D(h_0) = {\rm H}^2(\re), \quad h
: = h_0 + w, \quad D(h) = D(h_0).
$$
Assume that for some $\alpha > 0$ there exist real numbers
$\omega_{\pm}$ such that
    \bel{rr1}
\lim_{y \to \pm\infty} |y|^{\alpha} w(y) = \omega_{\pm}.
    \ee
In particular, \eqref{rr1} imples that $\sigma_{\rm ess}(h) = [0,\infty)$.\\ Set $\omega^{(-)}_{\pm} : = \max\{0,-\omega_{\pm}\}$,
$\omega^{(+)}_{\pm} : = \max\{0,\omega_{\pm}\}$.
    \begin{lemma}
\label{rrl1} {\rm \cite[Theorem XIII.82]{RS}} Assume that
\eqref{rr1} holds with $\alpha \in (0,2)$. Then we have
    \bel{rr2}
    \lim_{\lambda \downarrow 0} \lambda^{\frac{1}{2} -
    \frac{1}{\alpha}} N(-\lambda; h) = {\mathcal C}_{\alpha}
    \left(\left(\omega^{(-)}_{-}\right)^{1/\alpha} +
    \left(\omega^{(-)}_{+}\right)^{1/\alpha}\right).
    \ee
\end{lemma}
{\em Remark}: Under the hypotheses of Lemma \ref{rrl1}, relation
\eqref{rr2} is equivalent to the standard semiclassical formula
$$
N(-\lambda; h) = (2\pi)^{-1} \left|\left\{(y,\eta) \in T^*\re \, |
\, \eta^2 + w(y) < -\lambda\right\}\right|(1 + o(1)), \quad
\lambda \downarrow 0,
$$
where $|\cdot |$ denotes the Lebesgue measure, provided that $\omega^{(-)}_{-} + \omega^{(-)}_{+} > 0$.\\
Recall now that $\xi(-\lambda; h_q(\varepsilon), h_{0,q}) =
-N(-\lambda; h_{q}(\varepsilon))$, $\lambda > 0$. Since the
operator $h_{q}(\varepsilon)$ is unitarily equivalent to the
operator $h = h_0 + w$ with $w(y) =
w_{q,\varepsilon}(\mu_q^{1/2}y)$, and the quantities
$\omega_{q,\pm}(\varepsilon)$ are continuous at $\varepsilon = 0$,
we find that \eqref{m18} and \eqref{rr2} imply \eqref{a6}.
    \begin{lemma} \label{rrl2}
Assume that \eqref{rr1} holds with $\alpha \in (1,2)$. Then we
have
    $$
    \lim_{\lambda \downarrow 0} \lambda^{\frac{1}{2} -
    \frac{1}{\alpha}} \xi(\lambda;h,h_0) =
$$
\bel{rr4}
    - {\mathcal C}_{\alpha}
   \left(\csc{(\pi/\alpha)} \left(\left(\omega^{(-)}_{-}\right)^{1/\alpha} +
    \left(\omega^{(-)}_{+}\right)^{1/\alpha}\right) + \cot{(\pi/\alpha)} \left(\left(\omega^{(+)}_{-}\right)^{1/\alpha} +
    \left(\omega^{(+)}_{+}\right)^{1/\alpha}\right)\right).
    \ee
\end{lemma}
\begin{proof}
Set
$$
h^{(+)}_0 : = -\frac{d^2}{dy^2}, \quad D(h_0^{(+)}) = \left\{u \in {\rm
H}^2(0,\infty)\, | \, u(0) = 0 \right\},
$$
$$
 h^{(+)} : = h^{(+)}_0 +
w_{|(0,\infty)}, \quad D(h^{(+)}) = D(h^{(+)}_0).
$$
By the Birman-Krein formula \eqref{rr9}, and  \cite[Section 7,
Corollary]{Y1},
    \bel{rr5}
\lim_{\lambda \downarrow 0} \lambda^{\frac{1}{2} -
    \frac{1}{\alpha}} \xi(\lambda;h^{(+)},h^{(+)}_0) =    - {\mathcal C}_{\alpha}
   \left(\csc{(\pi/\alpha)}
    \left(\omega^{(-)}_{+}\right)^{1/\alpha} +
    \cot{(\pi/\alpha)}
    \left(\omega^{(+)}_{+}\right)^{1/\alpha}\right).
    \ee
    Now put
$$
h^{(-)}_0 : = -\frac{d^2}{dy^2}, \quad D(h_0^{(-)}) = \left\{u \in {\rm
H}^2(-\infty,0)\, | \, u(0) = 0 \right\},
$$
$$
 h^{(-)} : = h^{(-)}_0 +
w_{|(-\infty,0)}, \quad D(h^{(-)}) = D(h^{(-)}_0).
$$
Since the operator $h^{(-)}_0$ is unitarily equivalent to
$h^{(+)}_0$, and the operator $h^{(-)}$ is unitarily equivalent to
$h^{(+)}_0 + \tilde{w}$ with $\tilde{w}(y) = w(-y)$, $y>0$, we
find that \eqref{rr5} entails
    \bel{rr6}
    \lim_{\lambda \downarrow
0} \lambda^{\frac{1}{2} -
    \frac{1}{\alpha}} \xi(\lambda;h^{(-)},h^{(-)}_0) =    - {\mathcal C}_{\alpha}
   \left(\csc{(\pi/\alpha)}
    \left(\omega^{(-)}_{-}\right)^{1/\alpha} +
    \cot{(\pi/\alpha)}
    \left(\omega^{(+)}_{-}\right)^{1/\alpha}\right).
    \ee
    Making use of the orthogonal decomposition $L^2(\re) = L^2(-\infty,0) \oplus
    L^2(0,\infty)$, introduce the operators $h^{(-)}_0 \oplus
    h^{(+)}_0$ and $h^{(-)} \oplus
    h^{(+)}$, self-adjoint in $L^2(\re)$. Evidently,
    \bel{rr7}
    \xi(\lambda; h^{(-)} \oplus h^{(+)}, h^{(-)}_0 \oplus
    h^{(+)}_0)= \xi(\lambda; h^{(-)}, h^{(-)}_0) + \xi(\lambda; h^{(+)},
    h^{(+)}_0).
    \ee
Note that the resolvent differences
$
(h_0 - E_0)^{-1} - (h^{(-)}_0 \oplus
    h^{(+)}_0 - E_0)^{-1}$ and $(h - E_0)^{-1} - (h^{(-)} \oplus
    h^{(+)} - E_0)^{-1}$
    with $E_0 < \inf \sigma(h)$ are rank-one operators. This fact as well as the definition of the SSF for a pair of semibounded operators satisfying \eqref{10} (see e.g. \cite[Theorem 8.9.1]{Y2}), the chain rule for SSFs for trace-class perturbations
     (see e.g. \cite[Proposition 8.2.5]{Y2}), and Krein's estimate of the SSF for finite-rank perturbations  (see e.g. \cite[Theorem 8.2.1]{Y2}), imply
    \bel{rr8}
    \xi(\lambda; h, h_0) = \xi(\lambda; h^{(-)} \oplus h^{(+)}, h^{(-)}_0 \oplus
    h^{(+)}_0) + O(1), \quad \lambda > 0.
    \ee
    Now \eqref{rr4} follows from the combination of \eqref{rr8},
    \eqref{rr7}, and \eqref{rr5} -- \eqref{rr6}.
\end{proof}
The combination of \eqref{m18} and \eqref{rr4} easily yields \eqref{rr3}.\\

 {\bf 4.9.} Finally, we assume that $\alpha > 2$ and prove
Corollary \ref{f3}. If $\lambda < 0$, then \eqref{m31} is an
immediate consequence of Theorem \ref{t2} and the well-known fact
that the 1D Schr\"odinger operator $-\frac{d^2}{dy^2} + w(y)$, $y
\in \re$, has a most a finite number of negative eigenvalues if
$w(y) = o(|y|^{-2})$ as $|y| \to \infty$ (see e.g. \cite{RS}).
Assume $\lambda
> 0$. Combining \eqref{a20}, \eqref{a20a}, \eqref{m15}, and
\eqref{32a} with $p = 1$, we obtain
$$
|{\rm ind}\,(J-\varepsilon + {\rm Re}\,\tau_q(\lambda),
J-\varepsilon)| \leq n_*(1-|\varepsilon|; {\rm
Re}\,\tau_q(\lambda)) \leq
$$
    \bel{m30}
(1-|\varepsilon|)^{-1} \|M_{\phi, 1}(\lambda)\|_1 + {\rm
rank}\,M_{\phi, 2}(\lambda) + (1-|\varepsilon|)^{-1} \|M_{\phi,
3}(\lambda)\|_1, \quad \varepsilon \in (-1,1).
    \ee
Pick $\gamma < (\alpha-1)/2$, $\gamma \leq 1$, $\gamma > 1/2$.
Using \eqref{m4} -- \eqref{m33}, we find that the r.h.s. of
\eqref{m30} remains bounded as $\lambda \downarrow 0$.\\
Putting together  \eqref{a28}, \eqref{a27a} -- \eqref{a28a},
\eqref{m10} -- \eqref{m11}, and \eqref{m30}, we
obtain \eqref{m31} in the case $\lambda > 0$.\\

 {\bf Acknowledgements}.  The authors thank the referee whose valuable remarks
 contributed to the improvement of the article. Georgi Raikov gratefully acknowledges the financial  support by the
Chilean Scientific Foundation {\em Fondecyt} under Grant 1050716.
\\

\bigskip

{\sc Ph. Briet}\\
Centre de Physique Th\'eorique\\
CNRS-Luminy, Case 907\\
13288 Marseille, France\\
E-mail: briet@cpt.univ-mrs.fr\\

{\sc G. Raikov}\\
Facultad de Matem\'aticas\\
Pontificia Universidad Cat\'olica de Chile\\
Av. Vicu\~na Mackenna 4860\\ Santiago de Chile\\
E-mail: graikov@mat.puc.cl\\

{\sc E. Soccorsi}\\
Centre de Physique Th\'eorique\\
CNRS-Luminy, Case 907\\
13288 Marseille, France\\
E-mail: soccorsi@cpt.univ-mrs.fr\\

\end{document}